\documentclass[]{article}

\usepackage{amsmath, amssymb, dsfont, mathrsfs,amsthm}
\usepackage{float}
\usepackage{mdframed}
\usepackage{appendix}
\usepackage[hidelinks]{hyperref}
\usepackage{graphicx}
\usepackage{caption}
\usepackage[margin=1in]{geometry}

\allowdisplaybreaks

\title{Structural Results and Improved Upper Bounds on the Capacity of the Discrete-Time Poisson Channel
}

\author{Mahdi Cheraghchi \and Jo\~ao Ribeiro\thanks{
Department of Computing, Imperial College London, UK. Emails: \{m.cheraghchi, j.lourenco-ribeiro17\}@imperial.ac.uk.}
}

\date{}

\newcommand{\R}{\mathbb{R}}
\newcommand{\E}{\mathds{E}}

\newcommand{\eps}{\epsilon}

\newcommand{\DTP}{{\sf DTP}}

\newtheorem{thm}{Theorem}

\newtheorem{lem}[thm]{Lemma}

\newtheorem{remark}[thm]{Remark}

\newcommand{\Ch}{\mathsf{Ch}}

\newcommand{\Yq}{Y^{(q)}}

\newcommand{\KL}{D_\mathsf{KL}}

\newcommand{\OYmu}{\Omega_{Y,A,\mu}}
\newcommand{\OYeq}{\Omega^=_{Y,A,\mu}}

\newcommand{\OYfin}{\Omega_{Y,A,\mathsf{fin}}}

\let\originalleft\left
\let\originalright\right
\renewcommand{\left}{\mathopen{}\mathclose\bgroup\originalleft}
\renewcommand{\right}{\aftergroup\egroup\originalright}

\usepackage{url}
\begin{document}

\maketitle

\begin{abstract}
New capacity upper bounds are presented for the discrete-time Poisson channel with no dark current and an average-power constraint. These bounds are a simple consequence of techniques developed for the seemingly unrelated problem of upper bounding the capacity of binary deletion and repetition channels.
Previously, the best known capacity upper bound in the regime where the average-power constraint does not approach zero was due to Martinez (JOSA B, 2007), which is re-derived as a special case of the framework developed in this paper. Furthermore, this framework is carefully instantiated in order to obtain a closed-form bound that noticeably improves the result of Martinez everywhere. Finally, capacity-achieving distributions for the discrete-time Poisson channel are studied under an average-power constraint and/or a peak-power constraint and arbitrary dark current. In particular, it is shown that the support of the capacity-achieving distribution under an average-power constraint only must be countably infinite. This settles a conjecture of Shamai (IEE Proceedings I, 1990) in the affirmative. Previously, it was only known that the support must be unbounded.
\end{abstract}



\section{Introduction}\label{sec:intro}

We study the capacity of the classical discrete-time Poisson (DTP) channel, along with properties of its capacity-achieving distributions. Given an input $x \in \R^{\geq 0}$, the channel outputs a sample from Poisson distribution with mean $\lambda+x$, where $\lambda \geq 0$ is a channel parameter called the dark current. 
The DTP channel is motivated by applications in optical communication, involving a sender with a photon-emitting source and a receiver that observes the arrived photons (some of which may not have originated in the sender's source, hence the dark current parameter) \cite{Sha90}.

The capacity of the DTP channel is infinite if there are no constraints on the input distributions. For this reason, a power constraint should be imposed on the input distribution. The most typical choice, that we consider in this work, is an average-power constraint $\mu\in\R^{\geq 0}$, under which only input distributions $X$ satisfying $\E[X]\leq \mu$ are allowed. Several works also consider the case where a peak-power constraint is imposed on $X$, i.e., $X\leq A$ for some fixed $A\in\R^{>0}$ with probability 1 (e.g., \cite{LM09,LSVW11,WW14,SSEL15,AAGNKM15}). Setting $A=\infty$ corresponds to the case where no peak-power constraint is present.

Currently, no expression for the capacity of the DTP channel under an average-power constraint is known. Consequently, there has been considerable interest in obtaining sharp bounds and in determining the asymptotic behavior of the DTP channel capacity in several settings, and in investigating properties of capacity-achieving distributions. We focus on upper bounds for the capacity of the DTP channel with $\lambda=0$ under an average-power constraint $\mu$. Note that any such upper bound is also a capacity upper bound for the DTP channel with $\lambda>0$, as such a channel can be simulated from the DTP channel with $\lambda=0$ by having the receiver add an independent Poisson random variable with parameter $\lambda$ to the output.

The problem of better understanding the properties of capacity-achieving distributions for a given channel has also received significant attention. Normally, one is interested in determining whether a capacity-achieving distribution has finite or discrete support. Besides the fact that studying properties of such distributions may provide more insight into the channel capacity, it is also of practical importance. In fact, showing that the optimal distribution can be finite or discrete reduces the complexity of the problem of finding or approximating such a distribution, and allows the application of a wider range of numerical methods. The finiteness and discreteness of capacity-achieving distributions is well-understood for very general classes of noise-additive channels. However, much less is known for non-additive channels, and in particular the DTP channel.

\subsection{Previous work}
The two main regimes for studying the asymptotic behavior of the DTP channel capacity are when $\mu\to0$ and $\mu\to\infty$. Brady and Verd\'{u}~\cite{BV90} studied the asymptotic behavior of the capacity under an average-power constraint $\mu$ when $\mu\to\infty$ and $\mu/\lambda$ is kept fixed. Later, Lapidoth and Moser~\cite{LM09} studied the same problem when $\lambda$ is constant, with and without an additional peak-power constraint. When $\mu\to 0$, Lapidoth et al.\ \cite{LSVW11} determined the first-order asymptotic behavior of the capacity when $\mu$ goes to zero, both when $\mu/\lambda$ is kept constant and when $\lambda$ is fixed, with and without a peak-power constraint. Later, Wang and Wornell~\cite{WW14} improved their result when $\mu/\lambda$ is constant.

Obtaining capacity upper bounds for the DTP channel has been a major subject of interest. Explicit asymptotic capacity upper bounds for the DTP channel under an average-power constraint can be found in~\cite{LM09,LSVW11,WW14,AAGNKM15}. The current best non-asymptotic upper bound, which is in fact the best capacity upper bound outside the limiting case $\mu \to 0$, was derived by Martinez~\cite{Mar07}. However, its proof contains a small gap, as mentioned in~\cite{LM09}, and is not considered completely rigorous. A more detailed discussion of these upper bounds and of the asymptotic behavior of the capacity can be found in Section~\ref{sec:prevbounds}. While we focus on capacity upper bounds, we mention that explicit (asymptotic and non-asymptotic) capacity lower bounds for several settings have been derived in \cite{Mar07,LM09,CHC10,LSVW11,WW14,YZWD14}.  

There is a large amount of literature focusing on properties of capacity-achieving distributions for many classes of channels. As discussed before, one is mostly interested in determining whether such optimal distributions have finite or discrete support. The landscape of this problem is well-understood for quite general classes of noise-additive channels under several input constraints (see, e.g., the early works~\cite{Smi71,Sha90,AFTS01} and the recent works~\cite{ED18,FAF18,DGPS18})

The shape of capacity-achieving distributions for the DTP channel was first studied by Shamai~\cite{Sha90}, who showed that a capacity-achieving distribution for the DTP channel under a peak-power constraint must have finite support, and conjectured that the capacity-achieving distribution for the DTP channel under an average-power constraint only is discrete. He also gave conditions which ensure that distributions with two mass points are optimal. These results were extended by Cao, Hranilovic, and Chen \cite{CHC14a,CHC14b}. In particular, they showed that a capacity-achieving distribution for the DTP channel under an average-power constraint only must have unbounded support. Moreover, they also proved that such a distribution must have some mass at $x=0$, and, if a peak-power constraint $A$ is present, some mass at $x=A$ as well. Unlike noise-additive channels, not much is known about the capacity-achieving distributions of the DTP channel when there is only an average-power constraint present.

Other aspects and settings of the DTP channel have also received attention recently. A generalization of the DTP channel was studied by Aminian et al.\ \cite{AAGNKM15}, where simple and general capacity upper bounds in the presence of average- and peak-power constraints are also given for the classical DTP channel. Sutter et al.\ \cite{SSEL15} studied numerical algorithms for approximating the capacity of the DTP channel in the presence of both average- and peak-power constraints, and obtained sharp capacity bounds in this setting.

\subsection{Our contributions and techniques}
In the first part of this work, we derive improved capacity upper bounds for the DTP channel with $\lambda=0$ under an average-power constraint. Our technique is based on a natural convex duality formulation developed by Cheraghchi~\cite{Che17} for the seemingly unrelated problem of upper bounding the capacity of binary deletion and repetition channels. Furthermore, we prove new results on the shape of capacity-achieving distributions for the DTP channel.

We show that the result of Martinez~\cite{Mar07} can be obtained as an immediate special (sub-optimal) case of our results, thus giving a simple and rigorous proof for this bound. Furthermore, we extract two improved bounds from our more general result (Theorem~\ref{thm:mainbound}); one involving the minimization of a smooth convex function over $(0,1)$, as well as a closed-form bound (Theorem~\ref{thm:closedbound}). Both of these bounds are strictly tighter than the bound by Martinez for all $\mu>0$. Thus, we obtain the current best capacity upper bounds for the DTP channel with $\lambda=0$ under an average-power constraint $\mu$ for all values of $\mu$ outside the limiting case $\mu \to 0$. An additional feature of our results is that they are simple to derive.

In the second part, we study properties of capacity-achieving distributions for the DTP channel. Notably, we show that a capacity-achieving distribution for the DTP channel under an average-power constraint must be discrete. This settles a conjecture of Shamai~\cite{Sha90} in the affirmative. Previously, it was only known that the support was unbounded. In fact, we actually show the stronger result that the support must have finite intersection with all bounded intervals. This brings the state of knowledge on this topic for the DTP channel closer to that of noise-additive channels, which are much better understood. Our proof techniques are general and work under any dark current and any combination of average-power and peak-power constraints. In particular, we give an alternative proof that the capacity-achieving distribution under average- and peak-power constraints is finite, which was originally proved by Shamai~\cite{Sha90}.

The rest of the article is organized as follows: In Section~\ref{sec:prem} we introduce our notation. Further discussion of the best previously known bounds, along with the asymptotic behavior of the capacity when $\lambda=0$, appear in Section~\ref{sec:prevbounds}. The duality-based framework and the derivation of our upper bounds (including the bound by Martinez as a special case) are presented in Section~\ref{sec:mainbound}. Finally, we compare the bounds from Section~\ref{sec:mainbound} with those from Section~\ref{sec:prevbounds} in Section~\ref{sec:comp}. In Section~\ref{sec:shape}, we present our results on the shape of capacity-achieving distributions for the DTP channel.

\section{Notation}\label{sec:prem}

We denote the capacity of the DTP channel with average-power constraint $\mu$ and $\lambda=0$ by $C(\mu)$. We measure capacity in nats per channel use and denote the natural logarithm by $\log$. Random variables are usually denoted by uppercase letters such as $X$, $Y$, and $Z$. For a discrete random variable $X$, we denote by $X(x)$ the probability that $X$ takes on value $x$. The support of a random variable $X$ is denoted by $\mathsf{supp}(X)$, i.e., $\mathsf{supp}(X)$ is the smallest closed set $\mathcal{W}$ such that $\Pr[X\in \mathcal{W}]=1$. When the context is clear, we may at certain points confuse random variables and their associated cumulative distribution functions. The Kullback-Leibler divergence between $X$ and $Y$ is denoted by $D_{{\sf KL}}(X\|Y)$. In general, we use the convention that $0\log 0=0$.

\section{Previously known bounds and asymptotic results}\label{sec:prevbounds}

In this section, we survey the best previously known capacity upper bounds and the known results on the asymptotic behavior of $C(\mu)$. The asymptotic regimes considered in the literature are when $\mu\to 0$ and $\mu\to\infty$.

In the small $\mu$ regime, Lapidoth et al.~\cite{LSVW11} showed that
\begin{equation*}
\lim_{\mu\to 0}\frac{C(\mu)}{\mu\log(1/\mu)}=1.
\end{equation*}
Moreover, they gave the following upper bound matching the asymptotic behavior \cite[expression (86)]{LSVW11},
\begin{align}\label{bound:lapidoth}
C(\mu)&\leq -\mu\log p-\log(1-p)+\frac{\mu}{\beta}+\mu\cdot \max\bigg(0,\frac{1}{2}\log\beta+\log\left(\frac{\bar{\Gamma}(1/2,1/\beta)}{\sqrt{\pi}}+\frac{1}{2\beta}\right)\bigg),
\end{align}
where $p\in(0,1)$ and $\beta>0$ are free constants, and $\bar{\Gamma}$ is the upper incomplete gamma function. It is easy to see that the optimal choice for $p$ is $p=\frac{\mu}{1+\mu}$.

Later, Wang and Wornell \cite{WW14} determined the higher-order asymptotic behavior of $C(\mu)$ in the small $\mu$ regime, where it was shown that
\begin{equation*}
C(\mu)= \mu\log(1/\mu)-\mu\log\log(1/\mu)+O(\mu)
\end{equation*}
when $\mu\to 0$. This was previously noted by Chung, Guha, and Zheng \cite{CGZ11}, although they only proved the result for a more restricted set of input distributions (as mentioned in \cite{WW14}). Wang and Wornell \cite[expression (180)]{WW14} gave an upper bound (valid for small enough $\mu$) matching this asymptotic behavior; namely,
\begin{align}\label{bound:ww}
C(\mu)&\leq \mu+\mu\log\log\left(\frac{1}{\mu}\right)+\log\left(\frac{1}{1-\mu}\right)+\mu\log\left(\frac{1}{1-\frac{1}{\log(1/\mu)}}\right)+\mu\cdot \sup_{x\geq 0}\phi_\mu(x),
\end{align}
where $\phi_\mu(x):=\frac{1-e^{-x}}{x}\log\left(\frac{x}{\mu\log(1/\mu)}\right)$.

In the large $\mu$ regime, Lapidoth and Moser \cite{LM09} showed that
\begin{equation*}
\lim_{\mu\to \infty}\frac{C(\mu)}{\log \mu}=\frac{1}{2}.
\end{equation*}
The best upper bound in this regime (and, in fact, anywhere outside the asymptotic limit $\mu\to 0$) was derived by Martinez \cite[expression (10)]{Mar07} and is given by
\begin{align}\label{bound:martinez}
C(\mu)&\leq \left(\mu+\frac{1}{2}\right)\log\left(\mu+\frac{1}{2}\right)-\mu\log\mu-\frac{1}{2}+\log\left(1+\frac{\sqrt{2e}-1}{\sqrt{1+2\mu}}\right).
\end{align}
It holds that \eqref{bound:martinez} attains the first-order asymptotic behavior of $C(\mu)$ both when $\mu\to 0$ and when $\mu\to\infty$, and is strictly better than \eqref{bound:lapidoth} for all $\mu>0$. However, as noted in \cite{LM09}, the proof in \cite{Mar07} is not considered to be completely rigorous as it contains a gap (a certain equality is only shown numerically).

Aminian et al.\ \cite[Example 2]{AAGNKM15} give the upper bound
\begin{equation*}
	\sup_{X:\mathbb{E}[X]\leq \mu}\mathsf{Cov}(X+\lambda,\log(X+\lambda))
\end{equation*}
for the capacity of the DTP channel with an average-power constraint $\mu$ and dark current $\lambda$, where $\mathsf{Cov}(\cdot,\cdot)$ denotes the covariance. However, this bound is only useful when $\lambda$ is large.

Finally, we note that an analytical lower bound is also given in~\cite{Mar07}. This lower bound is obtained by considering gamma distributions as the input to the DTP channel (and thus negative binomial distributions as the corresponding output). More precisely, we have
\begin{align}
	C(\mu)&\geq (\mu+\nu)\log\left(\frac{\mu+\nu}{\nu}\right)+\mu(\psi(v+1)-1)\nonumber\\&-\int_0^1 \left(1-\left(\frac{\nu}{\nu+\mu(1-t)}\right)^\nu\right)\frac{t^{\nu-1}}{(1-t)\log t}-\frac{\mu}{\log t}dt\label{eq:martinezlb}
\end{align}
for all $\nu>0$, where $\psi(y)=\frac{d}{d y}\log\Gamma(y)$ is the digamma function ($\Gamma$ denotes the gamma function). Martinez~\cite{Mar07} also obtained the elementary lower bound $C(\mu)\geq \frac{1}{2}\log(1+\mu)$. These bounds behaves well when $\mu$ is large. In fact, the capacity is known to behave like $\frac{1}{2}\log\mu$ when $\mu\to \infty$.

\section{The proposed upper bounds}\label{sec:mainbound}

In this section, we derive new upper bounds on $C(\mu)$. While previous upper bounds are mostly based on duality results from \cite{LM03}, our derivation (although still duality based) follows from the application of a framework recently developed in \cite{Che17} in the context of binary deletion-type channels.

\subsection{The convex duality formulation}\label{sec:highlevel}

In this section, we give a high-level overview of our approach towards obtaining improved capacity upper bounds. 

We denote the DTP channel with dark current $\lambda$ under an \emph{output} average-power constraint $\mu$ and an input peak-power constraint $A$ by $\DTP_{\lambda,A,\mu}$. In words, this channel accepts input distributions $X$ such that $\mathsf{supp}(X)\subseteq [0,A]$ and which have associated output distributions $Y$ satisfying $\mathds{E}[Y]\leq \mu$. We may set $A=\infty$, in which case there is no peak-power constraint. When $A=\infty$ and $\lambda=0$, we denote the corresponding channel by $\DTP_\mu$. 

Note that imposing an average-power constraint $\mu$ on the output of the DTP channel is equivalent to imposing an average-power constraint $\mu-\lambda$ on its input. Because of this, one can easily move back and forth between input and output average-power constraints for the DTP channel. We may refer to ``input average-power constraint" simply as ``average-power constraint" throughout the paper.

A main component of our proofs is the following natural duality result for the DTP channel. This statement was originally proved for general channels with discrete input and output alphabets in \cite{Che17}. A proof of Theorem~\ref{thm:dual} for a general class of channels with continuous input under output average-power constraints and/or input peak-power constraints is presented in Appendix~\ref{sec:proof1}.
\begin{thm}[\protect{\cite[Theorem 1]{Che17}}, adapted]\label{thm:dual}
	
	Suppose there exist a random variable $Y$, supported on $\mathbb{N}$, and parameters $\nu_0\in\R$ and $\nu_1\in\mathbb{R}^{\geq 0}$ such that
	\begin{equation}\label{eq:KLreq}
	D_{{\sf KL}}(Y_x\| Y)\leq \nu_0+\nu_1\E[Y_x]
	\end{equation}
	for every $x\in [0,A]$, where $Y_x$ denotes the output of the DTP channel with dark current $\lambda$ when $x$ is given as input, i.e., $Y_x$ follows a Poisson distribution with mean $\lambda+x$. Assume that $\mu>\lambda$ (otherwise the problem is trivial). Then, we have
	\[
	C(\DTP_{\lambda,A,\mu})\leq\nu_0+ \nu_1\mu.
	\]
	Moreover, an input distribution $X$ is capacity-achieving for $\DTP_{\lambda,A,\mu}$ and
	\[
	C(\DTP_{\lambda,A,\mu})= \nu_0+\nu_1\mu
	\]
	if and only if its corresponding output distribution $Y$ satisfies $\mathds{E}[Y]=\mu$ and
	\[
	D_{{\sf KL}}(Y_x\| Y)\leq \nu_0+\nu_1\E[Y_x]
	\]
	for every $x\in[0,A]$, with equality for all $x\in{\sf supp}(X)$.
\end{thm}

Although there is a very simple correspondence between input and output average-power constraints for the DTP channel, this is not always the case. For general channels, considering the output mean as a parameter (as opposed to the input mean) leads to a more natural design of candidate distributions to be used in the analogue of Theorem~\ref{thm:dual}.

We call distributions $Y$ satisfying~\eqref{eq:KLreq} in Theorem~\ref{thm:dual} for some parameters $\nu_0$ and $\nu_1$ \emph{dual-feasible.} For the DTP channel with $\lambda=0$, we wish to find a dual-feasible distribution $Y$ and 
parameters $\nu_0, \nu_1 > 0$ such that
\[
D(Y_x||Y)\leq \nu_0+\nu_1 \E[Y_x]=\nu_0+\nu_1x
\]
for all $x\in\mathbb{R}^{\geq 0}$, and the inequality gap 
as small as possible. Using Theorem~\ref{thm:dual}, we readily obtain an upper bound for $C(\DTP_\mu)=C(\mu)$.

\subsection{The digamma distribution}\label{sec:digamma}

The result of Martinez~\cite{Mar07} follows the common approach of a 
convex duality formulation that leads to capacity upper bounds given
an appropriate distribution on the channel output alphabet. Indeed, this
is also the approach that we take. The dual distribution chosen by
\cite{Mar07} is a negative binomial distribution, which is a natural 
choice corresponding to a gamma distribution for the channel input. 
However, lengthy manipulations and certain adjustments are needed to obtain a closed-form
capacity upper bound for this choice. We use a slightly different
duality formulation, as discussed in Section~\ref{sec:highlevel}. Furthermore, for the dual output distribution, we use
a distribution that we call the ``digamma distribution'' and is
designed by Cheraghchi~\cite{Che17} precisely for the purpose
of use in the duality framework of \cite{Che17}. This distribution
asymptotically behaves like the negative binomial distribution. However,
it is constructed to automatically yield provable capacity upper bounds
without need for any further manipulations or adjustments. This is the key to our
refined bounds and dramatically simplified analysis\footnote{We note
that the duality framework of \cite{Che17} uses standard 
techniques and the dual-feasibility of the digamma distribution
also has a simple proof.}.

For a parameter $q\in(0,1)$, the digamma distribution $Y^{(q)}$ is defined
over non-negative integers with probability mass function
\begin{equation}\label{eq:digamma}
Y^{(q)}(y):=y_0 \frac{\exp(y \psi(y))(q/e)^y}{y!},\quad y=0,1,\ldots,
\end{equation}
where $y_0$ is a normalizing factor depending on $q$ (we omit this dependence in the notation for brevity), 
$\psi$ is the digamma function,
and $y \psi(y)$ is understood to be zero for $y=0$. For positive integers $y$, we have $\psi(y)=-\gamma+\sum_{k=1}^{y-1}1/k$, where $\gamma\approx 0.5772$ is the Euler-Mascheroni constant.

We will need to control the normalizing factor $y_0$, which is accomplished by the following result.
\begin{lem}[\protect{\cite[Corollary 16]{Che17}}]\label{lem:y0upper}
	We have
	\begin{align*}
		\log\left(1+\frac{2}{e^{1+\gamma}}\left(\frac{1}{\sqrt{1-q}}-1\right)\right)\leq-\log y_0\leq \log\left(1+\frac{1}{\sqrt{2e}}\left(\frac{1}{\sqrt{1-q}}-1\right)\right)
	\end{align*}
for all $q\in(0,1)$.
\end{lem}
\begin{remark}
	Sharper bounds exist for $-\log y_0$ based on special functions (Lerch transcendent).
\end{remark}

We will also be using the fact that the digamma distribution is closely related to the negative binomial distribution. We denote the negative binomial distribution with number of failures $r$ (note that $r$ is not necessarily an integer) and success probability $p$ by ${\sf NB}_{r,p}$. Its probability mass function is given by
\[
	\mathsf{NB}_{r,p}(y)=\binom{y+x-1}{x}p^y(1-p)^r,\quad y=0,1,2,\dots.
\]
We have the following result.

\begin{lem}[\protect{\cite[Corollary 16]{Che17}}]\label{lem:negbin}
	For all $y\geq 1$ and $q\in(0,1)$,
	\[
		\frac{2}{e^{1+\gamma}}{\sf NB}_{1/2,q}(y)\leq \frac{\sqrt{1-q}P_{Y^{(q)}}(y)}{y_0}\leq \frac{1}{\sqrt{2e}}{\sf NB}_{1/2,q}(y).
	\]
\end{lem}

\subsection{A first capacity upper bound}\label{sec:firstbound}

In this section, we use the digamma distribution and the approach outlined in Section~\ref{sec:highlevel} in order to derive an upper bound for $C(\mu)$. 


The random variable $Y_x$ in this case satisfies $Y_x=\mathsf{Poi}(x)$, where $\mathsf{Poi}(\lambda)$ denotes a Poisson distribution with mean $\lambda$. Therefore, its probability mass function is given by
\[
	Y_x(y)=e^{-x}\frac{x^y}{y!},\quad y=0,1,2,\dots.
\]

We will now give a short proof that the digamma distribution given in~\eqref{eq:digamma} is dual-feasible for the $\DTP_\mu$ channel by invoking well-known facts from the theory of special functions.

First, for $q\in(0,1)$ and some function $g$ satisfying $g(y)\leq y\log y+o(y)$, consider a general distribution $Y$ of the form
\[
	Y(y)=y_0\frac{\exp(g(y))(q/e)^y}{y!},\quad y=0,1,2,\dots,
\]
where $y_0$ is the normalizing factor. The upper bound on $g$ ensures that $Y$ is a valid probability distribution. In this case, the Kullback-Leibler divergence between $Y_x$ and $Y$ has a simple form for every $x$. We have
\begin{align}\label{eq:KL}
D_{\sf{KL}}\left(Y_x\big|\big|Y\right)&=\sum_{y=0}^\infty Y_x(y)\log\left(\frac{Y_x(y)}{Y(y)}\right)\nonumber\\
&= \sum_{y=0}^\infty Y_x(y)(-\log y_0+y(1-\log q )-g(y)-x+y\log x)\nonumber\\
&=-\log y_0-x\log q+x\log x-\E[g(Y_x)].
\end{align}

Via~\eqref{eq:KL}, it follows that $Y$ is dual feasible provided that we choose $g$ such that
\begin{equation}\label{eq:conddualfeasible}
	\E[g(Y_x)]=e^{-x}\sum_{y=0}^\infty \frac{g(y)}{y!}x^y\geq x\log x
\end{equation}
for all $x\geq 0$.

From the theory of special functions (by instantiating the Tricomi confluent hypergeometric function $U(a,n+1,z)$ with approriate parameters: \cite[13.1.6, p.\ 505 with $a=n+1=1$]{AS65} combined with \cite[13.6.12, p.\ 509]{AS65} and \cite[13.6.30, p.\ 510]{AS65}), we have the identity
\begin{equation}\label{eq:specfunc}
	e^x E_1(x)= \sum_{y=0}^\infty \frac{\psi(1+y)}{y!}x^y-e^x\log x,
\end{equation}
where $E_1(x)=\int_1^{\infty}e^{-xt} dt/t$ is the exponential integral function and $\psi$ is the digamma function. Multiplying both sides of~\eqref{eq:specfunc} by $xe^{-x}$ leads to
\[
	e^{-x}\sum_{y=0}^\infty \frac{y\psi(y)}{y!}x^y = x\log x +xE_1(x)\geq x\log x.
\]
Consequently, the choice
\begin{equation}\label{eq:choiceg}
	g(y)=y\psi(y)
\end{equation}
with the convention $g(0)=0$ satisfies~\eqref{eq:conddualfeasible} and thus leads to a dual feasible distribution $Y$. Furthermore, $g(y)=y\log y+o(y)$, as desired. 

We briefly give some intuition as to how~\eqref{eq:choiceg} shows naturally in~\cite{Che17}. A possible approach towards tightly satisfying~\eqref{eq:conddualfeasible} is to design $g^*$ such that
\[
	\mathds{E}[g^*(Y_x)]=x\log x,\quad \forall x\geq 0.
\]
It is possible to derive a formal solution $g^*$ to this functional equation of the form $g^*(y)=\int_0^\infty h(y,t)dt$ for some function $h(\cdot,\cdot)$. However, $g^*(y)$ is a divergent integral for all $y>0$. Therefore, $g^*$ does not exist. A possible solution to this problem is to truncate the integration bounds so that the integral converges. Using some identities from the theory of special functions, truncating the integration bounds of $g^*(y)$ appropriately leads to the choice~\eqref{eq:choiceg}.

Combining the choice of $g$ in~\eqref{eq:choiceg} with \eqref{eq:KL} allows us to conclude that
\begin{equation}\label{ineq:klgap}
	D_{\sf{KL}}\left(Y_x\big|\big|Y^{(q)}\right)\leq -\log y_0-x\log q
\end{equation}
for all $x\geq 0$. Applying Theorem~\ref{thm:dual}, we conclude that
\begin{equation}\label{bound:mean-limited}
	C(\mu)=C(\DTP_\mu)\leq -\log y_0-\mu\log q,
\end{equation}
%
which immediately leads to the following result.
\begin{thm}\label{thm:mainbound}
For all $\mu \geq 0$, we have
\begin{equation}\label{bound:our}
	C(\mu)\leq \inf_{q\in(0,1)} (-\log y_0-\mu\log q).
\end{equation}
\end{thm}


\subsection{Elementary bounds in a systematic way}\label{sec:qchoice}

While Theorem~\ref{thm:mainbound} gives an upper bound on $C(\mu)$, it involves minimizing a rather complicated function (for which we do not know an exact closed-form expression) over a bounded interval. Since it is of interest to have easy-to-compute but high quality upper bounds, we consider instantiating the parameter $q$ inside the infimum in \eqref{bound:our} with a simple function of $\mu$. In this section, we present a systematic way of deriving such a good choice $q(\mu)$. Finally, we upper bound $-\log(y_0)$ using Lemma~\ref{lem:y0upper}, obtaining an improved closed-form bound for $C(\mu)$.

We determine a good choice $q(\mu)$ for the parameter $q$ in \eqref{bound:our} indirectly by instead choosing $q(\mu)$ so that the associated distribution $Y^{(q(\mu))}$ (given by \eqref{eq:digamma}) has expected value close to $\mu$. The reasons for this are the following: First, a capacity-achieving distribution $X$ under an average-power constraint $\mu$ must satisfy $\E[X]=\E[Y]=\mu$ (see Appendix~\ref{sec:existoptimal}). While a capacity-achieving $X$ does not necessarily induce a digamma distribution over the output, the digamma distribution seems to be close to optimal, since the gap between the two expressions in~\eqref{ineq:klgap} is $xE_1(x)$, which decays exponentially with $x$. Second, numerical computation suggests that the distribution $Y$ induced by the choice of $q$ that minimizes the bound from Theorem~\ref{thm:mainbound} has expected value very close (or equal) to $\mu$. While determining a choice $q(\mu)$ such that $\E\left[Y^{(q(\mu))}\right]$ is very close to $\mu$ for all $\mu>0$ may be complicated, we settle for a choice $q(\mu)$ that behaves well when $\mu\to 0$ and $\mu\to\infty$.

We begin by studying how $q(\mu)$ should behave when $\mu\to \infty$. In this case, we should have $q(\mu)\to 1$. Lemma~\ref{lem:y0upper} implies that
\begin{equation*}
	\frac{2}{e^{1+\gamma}}+\left(1-\frac{2}{e^{1+\gamma}}\right)\sqrt{1-q}\leq \frac{\sqrt{1-q}}{y_0}\leq \frac{1}{\sqrt{2e}}+\left(1-\frac{1}{\sqrt{2e}}\right)\sqrt{1-q},
\end{equation*}
from which we can conclude that
\begin{equation}\label{eq:asympqy0}
	\frac{2}{e^{1+\gamma}}\leq \frac{\sqrt{1-q}}{y_0}\leq \frac{1}{\sqrt{2e}}+o(1)
\end{equation}
when $q\to 1$. Combining \eqref{eq:asympqy0} with Lemma~\ref{lem:negbin}, we obtain
\begin{equation*}
	\frac{2\sqrt{2e}}{e^{1+\gamma}}-o(1)\leq \frac{\Yq(y)}{{\sf NB}_{1/2,q}(y)}\leq \frac{e^{1+\gamma}}{2\sqrt{2e}}\approx 1.038
\end{equation*}
for $y=0,1,\dots$, when $q\to 1$, and so we conclude that the digamma distribution is well-approximated by ${\sf NB}_{1/2,q}$ when $q$ is close to 1. 


Recall that we want a choice of $q(\mu)$ such that $Y^{(q(\mu))}$ has expected value as close as possible to $\mu$ in the large $\mu$ regime. The choice of $q$ which ensures that $\E\left[{\sf NB}_{1/2,q}\right]=\mu$ is $q=\frac{2\mu}{1+2\mu}$, and so we want $q(\mu)$ to satisfy $q(\mu)=\frac{2\mu}{1+2\mu}+o\left(\frac{1}{\mu}\right)$ when $\mu\to\infty$.

One could set $q(\mu)=\frac{2\mu}{1+2\mu}$ to obtain the desired behavior above, but we will show that we can correct this choice in order to achieve $\E\left[Y^{(q(\mu))}\right]= \mu+o(\mu)$ when $\mu\to 0$. To make the derivation simpler, we will instead work with the quantity $\frac{1}{1-q(\mu)}$.

Consider a choice $q(\mu)$ satisfying
\begin{equation*}
\frac{1}{1-q(\mu)}=1+\alpha\mu+\frac{\beta\mu^2}{1+\mu}
\end{equation*}
for some constants $\alpha$ and $\beta$. It is easy to see that $\frac{1}{1-q(\mu)}$ behaves as $1+\alpha\mu+o(\mu)$ when $\mu\to 0$ and as $1+(\alpha+\beta)\mu+o(\mu)$ when $\mu\to\infty$, which means we can set its asymptotic behavior in both the small and large $\mu$ regimes independently of each other. Moreover, setting $\alpha+\beta=2$ leads to the desired behavior $q(\mu)=\frac{2\mu}{1+2\mu}+o\left(\frac{1}{\mu}\right)$ when $\mu\to\infty$.

We now proceed to choose $\alpha$. As mentioned before, we determine the choice of $\alpha$ which ensures that $\E\left[Y^{(q(\mu))}\right]= \mu+o(\mu)$ when $\mu\to 0$. It is straightforward to see that, by construction, $q(\mu)=\alpha\mu+o(\mu)$ when $\mu\to 0$. We will need the following result.
\begin{lem}\label{lem:asympmean}
	We have $\E\left[Y^{(q)}\right]= e^{-(1+\gamma)}q+o(q)$ as $q\to 0$.
\end{lem}
\begin{proof}
	Recall that $g(y)=y\psi(y)$, and note that
	\begin{equation}\label{eq:expval}
		\frac{\E\left[Y^{(q)}\right]}{q}=y_0e^{-(1+\gamma)}+y_0\sum_{y=2}^\infty y\cdot \frac{e^{g(y)-y}q^{y-1}}{y!}.
	\end{equation}
	It is easy to see that $y_0$ approaches $1$ (using Lemma~\ref{lem:y0upper}, for example) and the second term in the RHS of \eqref{eq:expval} vanishes when $q\to 0$, and so the result follows.
\end{proof}
The remarks above, combined with Lemma~\ref{lem:asympmean}, imply that $\E\left[Y^{(q(\mu))}\right]= e^{-(1+\gamma)}\alpha\mu+o(\mu)$ when $\mu\to 0$. Therefore, it suffices to set $\alpha=e^{1+\gamma}$ to have $\E\left[Y^{(q(\mu))}\right]=\mu+o(\mu)$ when $\mu\to 0$. Based on this, we set $q(\mu)$ to be such that
\begin{equation}\label{eq:qchoice}
	\frac{1}{1-q(\mu)}=1+e^{1+\gamma}\mu+\frac{(2-e^{1+\gamma})\mu^2}{1+\mu}.
\end{equation}

Combining the previous discussion, Theorem~\ref{thm:mainbound}, and Lemma~\ref{lem:y0upper}, we obtain the following result.
\begin{thm}\label{thm:closedbound}
We have
\begin{equation}\label{bound:optquppery0}
	C(\mu)\leq \inf_{q\in(0,1)} f(\mu,q),
\end{equation}
where $	f(\mu,q):=-\mu\log q+\log\left(1+\frac{1}{\sqrt{2e}}\left(\frac{1}{\sqrt{1-q}}-1\right)\right)$.

In particular, by instantiating $q$ with $q(\mu)$ defined in \eqref{eq:qchoice},
\begin{align}\label{bound:qchoice}
	C(\mu)&\leq \mu  \log \left(\frac{1+\left(1+e^{1+\gamma }\right) \mu+2\mu ^2}{e^{1+\gamma } \mu + 2 \mu ^2}\right)+\log \left(1+\frac{1}{\sqrt{2e}}\left(\sqrt{\frac{1+(1+e^{1+\gamma})\mu+2\mu^2}{1+\mu}}-1\right)\right).
\end{align}
\end{thm}

Note that $f(\mu,\cdot)$ is an elementary, smooth, and convex function for every fixed $\mu\geq 0$. Therefore, \eqref{bound:optquppery0} can be easily approximated to any desired degree of accuracy.

\begin{remark}
	The reasons why we base our choice of $q(\mu)$ on \eqref{bound:our} instead of \eqref{bound:optquppery0} are the following: First, $q(\mu)$ is still close to optimal when used in \eqref{bound:optquppery0} (see Figure~\ref{fig:allbounds02}). Second, the choice is independent of the upper bound on $-\log y_0$, and so can be reutilized if a better bound is used.
\end{remark}

\subsection{The result of Martinez as a special case}\label{sec:martinez}
In this section, we show that the bound by Martinez~\eqref{bound:martinez} can be quite easily recovered through our techniques. More precisely, we show that this bound is a special case of \eqref{bound:optquppery0} with a sub-optimal choice of $q=2\mu/(1+2\mu)$. In particular, this implies that \eqref{bound:optquppery0} is strictly tighter than~\eqref{bound:martinez}. In this section, we define $m(\mu)$ to be the right hand side of \eqref{bound:martinez}. Recall that $	f(\mu,q)=-\mu\log q+\log\left(1+\frac{1}{\sqrt{2e}}\left(\frac{1}{\sqrt{1-q}}-1\right)\right)$.
\begin{thm}\label{thm:martinez}
	We have $f\left(\mu,\frac{2\mu}{1+2\mu}\right)=m(\mu)$ for all $\mu\geq 0$. Moreover, for every $\mu>0$ there is $q^*_\mu\in(0,1)$ such that $f(\mu,q^*_\mu)<m(\mu)$.
\end{thm}
\begin{proof}
	To prove the first statement of the theorem, we compute 
	\begin{align*}
	&m(\mu)-f\!\left(\mu,\frac{2\mu}{1+2\mu}\right)\\
	&=\left(\mu+\frac{1}{2}\right)\log\left(\mu+\frac{1}{2}\right)-\mu\log\mu-\frac{1}{2}+\log\left(1+\frac{\sqrt{2e}-1}{\sqrt{1+2\mu}}\right) \\&- \mu\log\left(\frac{1+2\mu}{2\mu}\right)-\log\left(1+\frac{1}{\sqrt{2e}}\left(\sqrt{1+2\mu}-1\right)\right)\\
	&=\frac{1}{2}\log\left(\mu+\frac{1}{2}\right)+\mu\left(\log\left(\mu+\frac{1}{2}\right)-\log\mu\right)-\frac{1}{2}+\log\left(\frac{\sqrt{1+2\mu}+\sqrt{2e}-1}{\sqrt{1+2\mu}}\right) \\&- \mu\log\left(\frac{1+2\mu}{2\mu}\right)-\log\left(\frac{\sqrt{1+2\mu}+\sqrt{2e}-1}{\sqrt{2e}}\right)\\
	&=\frac{1}{2}\log\left(\mu+\frac{1}{2}\right)-\frac{1}{2}+\log\left(\sqrt{\frac{2e}{1+2\mu}}\right)=0.
	\end{align*}
	
	To see that the second statement holds, it suffices to show that $\frac{\partial f}{\partial q}\left(\mu,\frac{2\mu}{1+2\mu}\right)\neq 0$ for all $\mu>0$. We have
	\begin{equation}\label{eq:deriv}
		\frac{\partial f}{\partial q}(\mu,q)=-\frac{\mu}{q}+\frac{1}{2\left(\sqrt{2e}+\frac{1}{\sqrt{1-q}}-1\right)(1-q)^{3/2}}.
	\end{equation}
	Instantiating with $q=\frac{2\mu}{1+2\mu}$ yields
	\begin{equation*}
	\frac{\partial f}{\partial q}\left(\mu,\frac{2\mu}{1+2\mu}\right)=-\frac{1+2\mu}{2}+\frac{(1+2\mu)^{3/2}}{2\left(\sqrt{2e}+\sqrt{1+2\mu}-1\right)},
	\end{equation*}
	and now it is enough to note that
	\begin{align*}
	& -(1+2\mu)\left(\sqrt{2e}+\sqrt{1+2\mu}-1\right)+(1+2\mu)^{3/2}\\
	& =(1+2\mu)(1-\sqrt{2e})< 0
	\end{align*}
	for all $\mu\geq 0$.
\end{proof}

Finally, we show that the explicit choice $q(\mu)$ from Section~\ref{sec:qchoice} yields a strictly better upper bound than the Martinez bound~\eqref{bound:martinez}.
\begin{thm}
	We have $f(\mu,q(\mu))<f\left(\mu,\frac{2\mu}{1+2\mu}\right)=m(\mu)$ for all $\mu>0$.
\end{thm}
\begin{proof}
We give short arguments that the statement holds whenever $\mu\geq 1.61$ and when $\mu$ is sufficiently small. The middle region can be verified numerically. Let $d(\mu):=m(\mu)-f(\mu,q(\mu))$.

We begin by noting that $q(\mu)>\frac{2\mu}{1+2\mu}$ for all $\mu>0$. This is because
\[
	q(\mu)-\frac{2\mu}{1+2\mu}=\frac{(e^{1+\gamma}-2)\mu}{(1+2\mu)(1+\mu(1+e^{1+\gamma}+2\mu))}>0.
\]
The inequality is true since the both the denominator and numerator are positive for all $\mu>0$ (observe that $e^{1+\gamma}>2$).

Due to the convexity of $f(\mu,\cdot)$, the desired statement holds for a given $\mu$ if $\frac{\partial f}{\partial q}(\mu,q(\mu))<0$. Recalling~\eqref{eq:deriv} and~\eqref{eq:qchoice}, we have
\begin{equation}\label{eq:dfdq}
	\frac{\partial f}{\partial q}(\mu,q(\mu))=-\mu-\frac{1+\mu}{e^{1+\gamma}+2\mu}+\frac{\left(\frac{\mu \left(2 \mu+e^{1+\gamma }+1\right)+1}{\mu+1}\right)^{3/2}}{\sqrt{\frac{\mu\left(2 \mu+e^{1+\gamma }+1\right)+1}{\mu+1}}+\sqrt{2 e}-1}.
\end{equation}
Let $T_1(\mu)=\mu+\frac{1+\mu}{e^{1+\gamma}+2\mu}$ and $T_2(\mu)=\frac{\mu \left(2 \mu+e^{1+\gamma }+1\right)+1}{\mu+1}$. Then, taking into account~\eqref{eq:dfdq}, we have $\frac{\partial f}{\partial q}(\mu,q(\mu))<0$ if and only if
\begin{align*}
	\frac{T_2(\mu)^{3/2}}{\sqrt{T_2(\mu)}+\sqrt{2e}-1}&<T_1(\mu),
\end{align*}
which is equivalent to
\begin{align*}
T_2(\mu)^{3/2}-T_1(\mu)\sqrt{T_2(\mu)}<T_1(\mu)(\sqrt{2e}-1.)
\end{align*}
Squaring both sides shows that $\frac{\partial f}{\partial q}(\mu,q(\mu))<0$ whenever
\begin{align*}
	T_2^3(\mu)-2T_1(\mu)T_2^2(\mu)-T_1(\mu)^2T_2(\mu)^2<T_1(\mu)^2(\sqrt{2e}-1)^2.
\end{align*}
Observe that each term in the inequality is a rational function. As such, we can expand each term, and then compute and eliminate the common denominator to obtain an equivalent polynomial inequality, which turns out to be
\begin{align*}
	e^{2+2 \gamma }-4 e^{1+\gamma }+8 \sqrt{2 e}-8 e+\left(24 \sqrt{2 e}-3 e^{2+2 \gamma }+e^{3+3 \gamma }-24 e-8\right) \mu\\+\left(24 \sqrt{2 e}-8 e^{1+\gamma }+2 e^{2+2 \gamma }-24 e-4\right) \mu^2+\left(8 \sqrt{2 e}-8 e-4\right) \mu^3<0.
\end{align*}
There are many known methods for determining the roots of degree-3 polynomials. We can use such a method to see that the largest root of the polynomial on the left-hand side is smaller than $1.61$, and so the inequality holds whenever $\mu\geq 1.61$.


To prove that $d(\mu)>0$ for $\mu$ small enough, we look at the limiting behavior of $d(\mu)$ when $\mu\to 0$. We have that
\begin{align*}
	d(\mu)=\left(1+\gamma+\frac{1}{\sqrt{2 e}}-\log 2 -\frac{e^{\frac{1}{2}+\gamma }}{2 \sqrt{2}}\right)\mu+o(\mu)\approx 0.27\mu+o(\mu)
\end{align*}
when $\mu\to 0$, which implies that $d(\mu)>0$, and hence $m(\mu)>f(\mu,q(\mu))$, when $\mu$ is small enough.


When $\mu<1.61$ but it is not too small, one can show $d(\mu)>0$ by employing a computer algebra system.  However, $d(\mu)$ is a complex expression, and so cannot be processed directly by such a system. We avoid this issue in the following way: For $\mu\in[0.3,1.61]$, we lower bound $d(\mu)$ by positive rational functions. This is done by replacing the logarithmic and square root terms of the expression by appropriate bounds (described below) which are themselves rational functions. Then, the question of whether $d(\mu)>0$ is reduced to showing that a certain polynomial is positive in the given interval, which can be formally checked by a computer algebra system with little effort. For $\mu<0.3$, our lower bounds for $d(\mu)$ are not good enough, and so we use the same reasoning to show that its second derivative $d''(\mu)$ is negative for $\mu<0.3$. This implies that $d(\mu)$ is concave in $[0,0.3]$, which, combined with the previous results, concludes the proof.

We do not explicitly write down the relevant lower bounds for $d(\mu)$ and upper bounds for the second derivative, as they feature high-degree polynomials. Instead, we describe the relevant bounds on the logarithmic and square root terms. Then, determining the corresponding rational function and formally checking whether it is positive in a given interval is a straightforward process.

The expression $d(\mu)$ features logarithmic terms, along with square root terms of the form $\sqrt{1+2\mu}$ and $\sqrt{(1 + (1 + e^{1+\gamma}) \mu + 2 \mu^2)/(1 + \mu)}$ (recall \eqref{bound:martinez} and \eqref{bound:qchoice}). For every $x\geq 1$, we have the bounds~\cite{Top06}
\begin{equation*}
	\frac{(x-1)(6+5(x-1))}{2(3+2(x-1))}\leq \log x\leq \frac{(x-1)(x+5)}{2x(2+x)}.
\end{equation*}
Furthermore, we can upper bound $\sqrt{1+2\mu}$ and $\sqrt{(1 + (1 + e^{1+\gamma}) \mu + 2 \mu^2)/(1 + \mu)}$ by their Taylor series of degree 5 and 3, respectively, around $\mu=1$. Replacing the relevant terms in $d(\mu)$ by their respective bounds described above yields a rational function lower bound which can be easily shown to be positive for $\mu\in[0.3,1.61]$ by a standard computer algebra system.

For $\mu<0.3$, the bounds above are not tight enough to show that $d(\mu)$ is positive, and so we focus on its second derivative $d''(\mu)$. However, $d''(\mu)$ cannot be processed directly by a computer algebra system either, and so we follow the same reasoning as before. The only terms of $d''(\mu)$ that need to be bounded are of the form $\sqrt{1+2\mu}$ and $\sqrt{(1 + (1 + e^{1+\gamma}) \mu + 2 \mu^2)/(1 + \mu)}$. It suffices to upper bound (resp., lower bound) $\sqrt{1+2\mu}$ by its Taylor series of degree 1 (resp., 2) around $\mu=0$. However, extra care is needed when dealing with $\sqrt{(1 + (1 + e^{1+\gamma}) \mu + 2 \mu^2)/(1 + \mu)}$. We split the interval $[0,0.3]$ into two intervals: First, in $(0,0.25]$ we lower bound the term by its Taylor series of degree 2 around $\mu=0$. Second, in $(0.25,0.3]$ we lower bound it by its Taylor series of degree 2 around $\mu=0.25$. 

Replacing the relevant terms of $d''(\mu)$ by their respective bounds, we obtain a negative rational function upper bounding $d''(\mu)$ in each of $(0,0.25]$ and $(0.25,0.3]$, which can be formally checked to be negative with a computer algebra system. This implies that $d(\mu)$ is concave in $(0,0.3]$, and so, combined with the facts that $d(\mu)>0$ for $\mu$ small enough and $d(\mu)>0$ for $\mu\geq 0.3$, we conclude that $d(\mu)>0$ for all $\mu>0$.
\end{proof}

\section{Comparison with previously known upper bounds}\label{sec:comp}

In this section, we compare the bounds from Theorem~\ref{thm:closedbound} with the previously known bounds described in Section~\ref{sec:intro}. Moreover, we investigate the loss incurred by using \eqref{bound:qchoice} instead of \eqref{bound:our}.

Figure \ref{fig:allbounds02} showcases a plot comparing the bounds from Theorem~\ref{thm:closedbound} to previously known bounds. The curve corresponding to the bound of Lapidoth et al.\ \eqref{bound:lapidoth} is actually the plot of $\mu\log\left(\frac{1+\mu}{\mu}\right)+\log(1+\mu)$, which lower bounds the RHS of \eqref{bound:lapidoth}. There is a noticeable improvement over the Martinez bound \eqref{bound:martinez} when $\mu$ is not very small, and one can see that \eqref{bound:qchoice} is very close to \eqref{bound:optquppery0} and~\eqref{bound:our}  (with significant overlap), which confirms that the choice $q(\mu)$ from Section~\ref{sec:qchoice} is close to optimal. Table~\ref{table:compbounds} gives the numerical values attained by \eqref{bound:our}, \eqref{bound:qchoice}, and \eqref{bound:martinez} for several values of $\mu$. Table~\ref{table:compq} compares the choice~\eqref{eq:qchoice} for $q(\mu)$ with the actual optimal value of $q$ for several values of $\mu$. As expected from the previous observations, the explicit choice is always quite close to the optimal value.

Due to the fact that our bounds are tighter than Martinez's bound, both of them satisfy the first-order asymptotic behavior of $C(\mu)$ when $\mu\to 0$ and when $\mu\to\infty$. However, they do not exhibit the correct second order asymptotic term when $\mu\to 0$. In fact, the second-order asymptotic term of our bounds when $\mu\to 0$ is $-O(\mu)$, while the correct term is $-\mu\log\log(1/\mu)$. For this reason, our bounds do not improve on the Wang-Wornell bound~\eqref{bound:ww} when $\mu$ is sufficiently small (numerically, when $\mu<10^{-6}$), while they noticeably improve on every previous bound when $\mu$ is not too small.

\begin{figure}
	\centering
	\captionsetup{justification=centering}
	\includegraphics[width=0.6\textwidth]{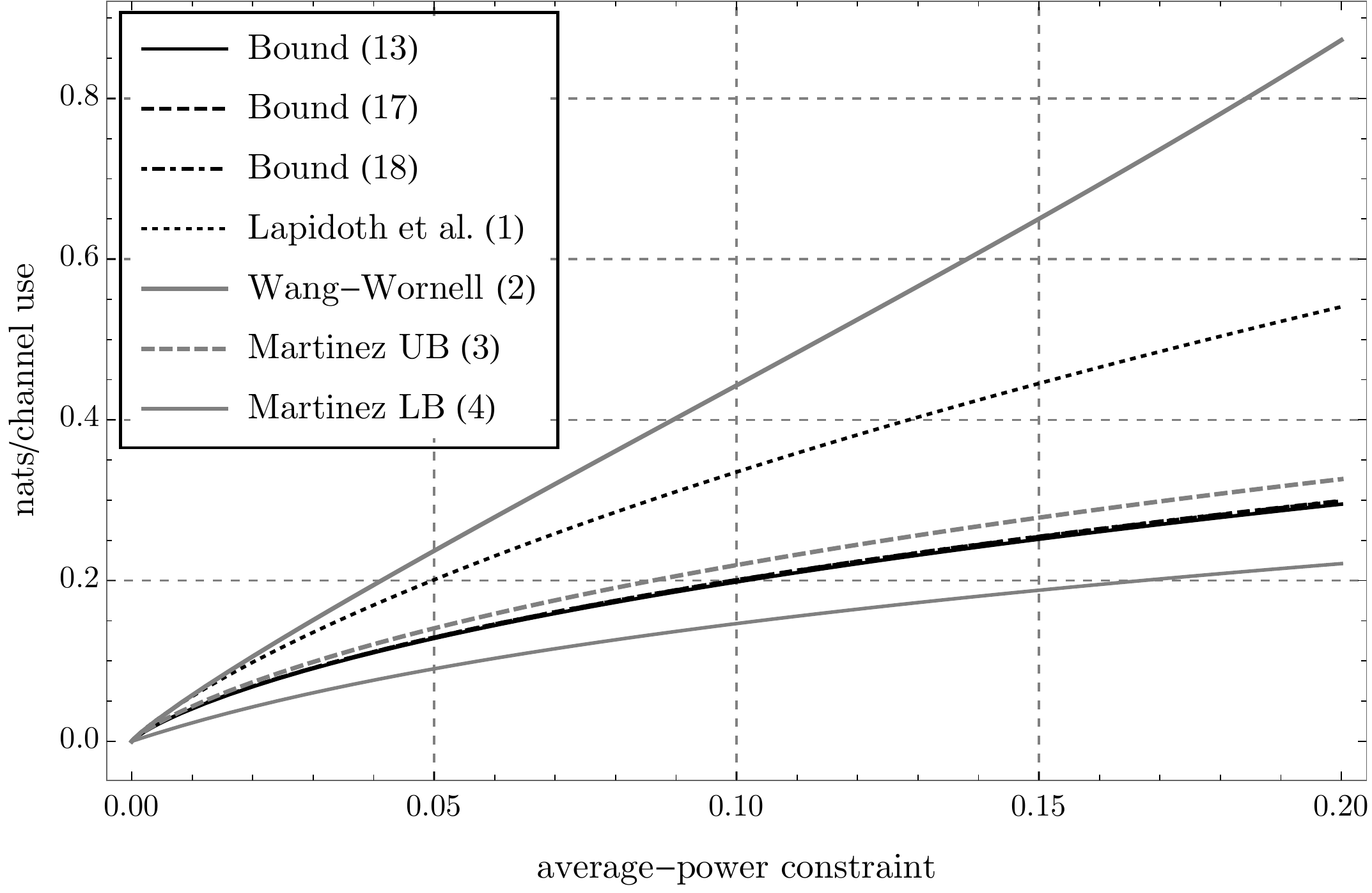}
	\caption{Comparison of upper bounds and the analytical lower bound~\eqref{eq:martinezlb} with $\nu=0.05$ for $\mu\in[0,0.2]$.}
	\label{fig:allbounds02}
\end{figure}

\begin{table}[h]
	\caption{Comparison between the bound~\eqref{bound:our} and the elementary bounds~\eqref{bound:qchoice} and~\eqref{bound:martinez} in nats/channel use.}\label{table:compbounds}
	\centering
	\begin{tabular}{ |c|c|c|c| }
		\hline
		$\mu$ & Bound \eqref{bound:our}& Bound \eqref{bound:qchoice} & Bound \eqref{bound:martinez}\\ 
		\hline
		0.05 & 0.1280& 0.1296 &0.1406\\ 
		\hline
		0.1 & 0.1983 & 0.2010 &0.2193\\ 
		\hline
		0.2 & 0.2951 &  0.2994 & 0.3262\\ 
		\hline
		0.5 & 0.4689 & 0.4753&0.5101\\ 
		\hline
		1 & 0.6367 & 0.6437 &0.6785\\ 
		\hline
		5 & 1.1407 & 1.1492 &1.1665\\ 
		\hline
		10 & 1.4005 & 1.4093 &1.4187\\ 
		\hline
		20 & 1.6806 & 1.6886 &1.6930\\ 
		\hline
		50 & 2.0756 & 2.0815 &2.0829\\ 
		\hline
	\end{tabular}
\end{table}

\begin{table}[h]
	\caption{Comparison between optimal $q$ in~\eqref{bound:our} for each $\mu$ and the choice $q(\mu)$ as in~\eqref{eq:qchoice}.}\label{table:compq}
	\centering
	\begin{tabular}{ |c|c|c| }
		\hline
		$\mu$ & Optimal $q$& $q(\mu)$ as in \eqref{eq:qchoice}\\ 
		\hline
		0.05 & 0.1851& 0.1905\\
		\hline
		0.1 & 0.3025 & 0.3143 \\
		\hline
		0.2 & 0.4482 &  0.4663 \\
		\hline
		0.5 & 0.6447 & 0.6607\\
		\hline
		1 & 0.7676 & 0.7738\\
		\hline
		5 & 0.9309 & 0.9252 \\
		\hline
		10 & 0.9617 & 0.9576 \\
		\hline
		20 & 0.9794 & 0.9771 \\
		\hline
		50 & 0.9912 & 0.9904 \\
		\hline
	\end{tabular}
\end{table}

Figure~\ref{fig:compmartinezdiff} showcases the distance of Martinez's bound~\eqref{bound:martinez} to \eqref{bound:our} and \eqref{bound:qchoice}. The plotted curves have similar shapes and are close to each other, which again shows that we do not lose much by replacing $-\log y_0$ by the upper bound of Lemma~\ref{lem:y0upper} and instantiating $q$ with the sub-optimal explicit choice $q(\mu)$ from Section~\ref{sec:qchoice}.

\begin{figure}
	\centering
	\captionsetup{justification=centering}
	\includegraphics[width=0.6\textwidth]{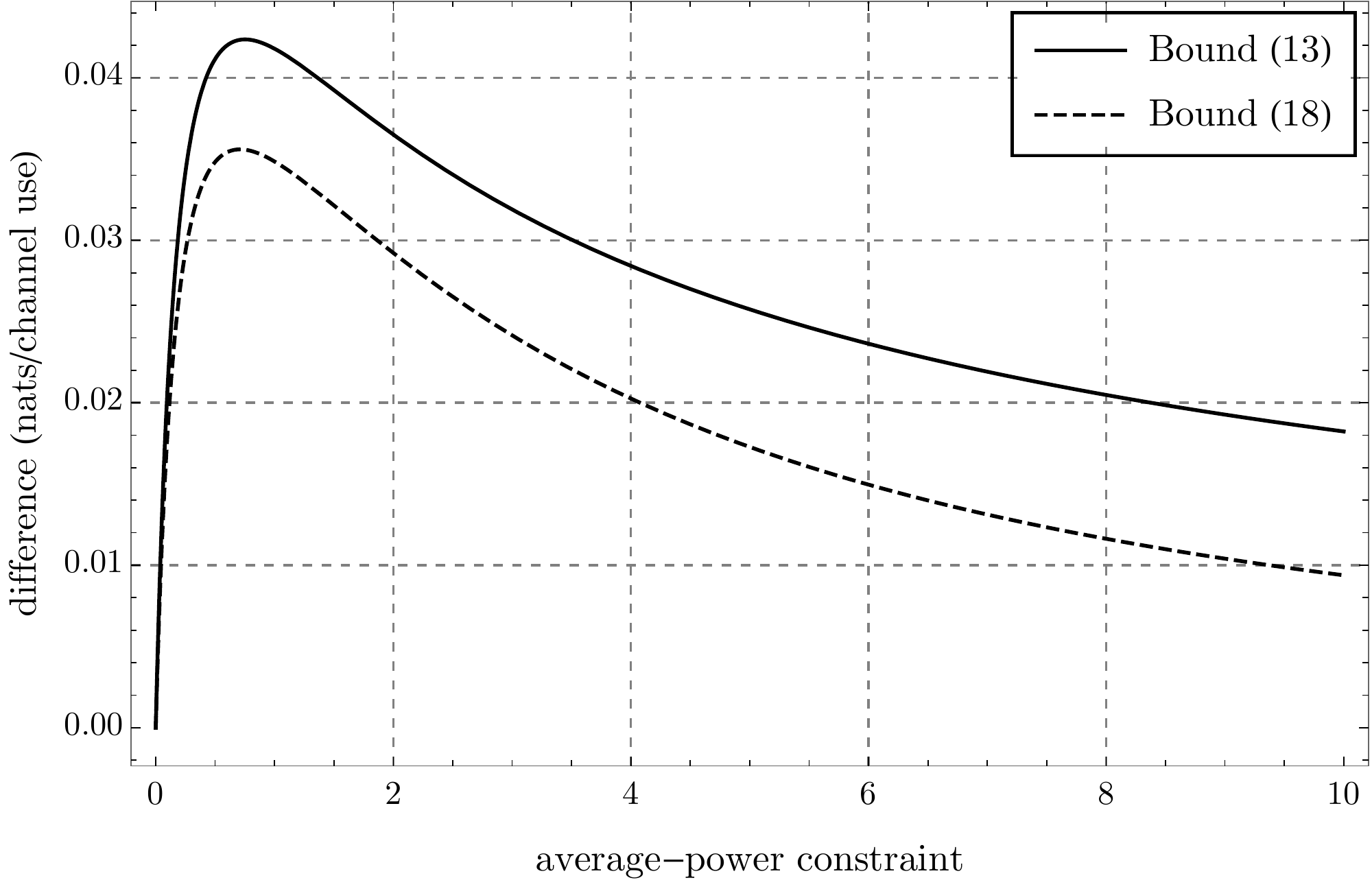}
	\caption{Comparison of difference between \eqref{bound:martinez} and \eqref{bound:our}, and between \eqref{bound:martinez} and \eqref{bound:qchoice} for $\mu\in[0,10]$.}
	\label{fig:compmartinezdiff}
\end{figure}

Figure~\ref{fig:relimprov} showcases the relative distance of the Martinez bound~\eqref{bound:martinez} to \eqref{bound:our} and \eqref{bound:qchoice}. In other words, if $m(\cdot)$ denotes the Martinez bound~\eqref{bound:martinez} and $b(\cdot)$ is either the RHS of \eqref{bound:our} or of \eqref{bound:qchoice}, then the plot shows the quantity $(m(\mu)-b(\mu))/m(\mu)$. Observe that, using \eqref{bound:qchoice}, we obtain an improvement of up to $8.2\%$ over \eqref{bound:martinez}, while we can get improvements close to $9.5\%$ using \eqref{bound:our}. Note that the two curves are close to each other and similar shape, reinforcing the fact that the loss incurred by using \eqref{bound:qchoice} instead of \eqref{bound:our} is small.
\begin{figure}
	\centering
	\captionsetup{justification=centering}
	\includegraphics[width=0.6\textwidth]{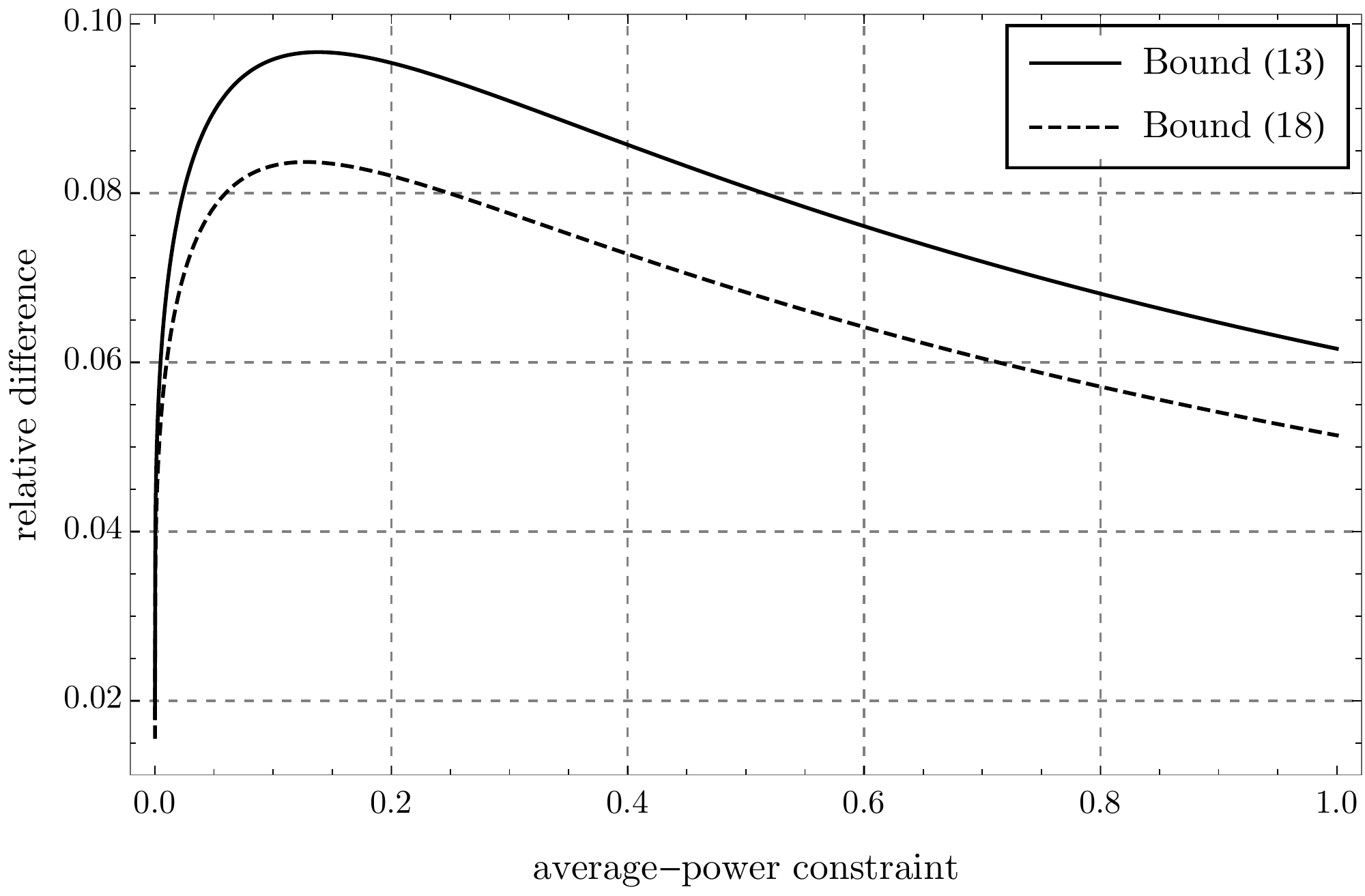}
	\caption{Relative difference between \eqref{bound:martinez} and \eqref{bound:our}, and between \eqref{bound:martinez} and \eqref{bound:qchoice} for $\mu\in[0,1]$.}
	\label{fig:relimprov}
\end{figure}

%
%
%

\section{The shape of capacity-achieving distributions}\label{sec:shape}

Besides understanding the capacity of communications channels, there has also been a significant amount of work towards determining the properties of capacity-achieving distributions. In particular, one is normally interested in knowing whether a capacity-achieving distribution has finite or discrete support, even though the input alphabet may not be a discrete set.

The study of capacity-achieving distributions for the DTP channel was initiated by Shamai~\cite{Sha90}, who proved that capacity-achieving distributions for the DTP channel with both average- and peak-power constraints have finite support. More recently, Cao, Hranilovic, and Chen~\cite{CHC14a,CHC14b} derived more properties of such distributions. Notably, they show that a capacity-achieving distribution must be supported at $0$ and at $A$ if a peak-power constraint $X\leq A$ is present. Furthermore, they show that distributions with bounded support are not capacity-achieving for the DTP channel with only an average-power constraint. For completeness, we show that there exist capacity-achieving distributions for the DTP channel under an average-power constraint in Appendix~\ref{sec:existoptimal}.

In this section, we show that a capacity-achieving distribution for the DTP channel with arbitrary dark current $\lambda\geq 0$ under an average-power constraint and/or a peak-power constraint must be discrete. As mentioned before, this settles a conjecture of Shamai~\cite{Sha90}. In fact, we show the stronger result that the support of a capacity-achieving distribution $X$ for the DTP channel under an average-power constraint and/or a peak-power constraint must have finite intersection with every bounded interval. Our techniques are general, and we recover Shamai's original result~\cite{Sha90} for the DTP channel under a peak-power constraint ($A<\infty$) with an alternative proof.

Consider a discrete probability distribution $Y$ supported on the non-negative integers. For our results, it suffices to consider $Y$ with full support. This is because all optimal output distributions of the DTP channel have full support. In fact, the only input distribution which does not induce an output distribution with full support is the distribution which assigns probability 1 to $x=0$, which is clearly not optimal. The following result gives a characterization of optimal output distributions for the DTP channel (which we might also call \emph{capacity-achieving} at times) that will be useful in later proofs.

\begin{lem}\label{lem:specform}
	Consider a distribution $Y$ with full support over the non-negative integers. Furthermore, for a given function $g$ define its (real-valued) exponential generating function $G$ as
	\begin{equation}\label{eq:expgen}
		G(z)=\sum_{i=0}^\infty \frac{g(i)}{i!}z^i.
	\end{equation}
	Let $Y_x=\mathsf{Poi}(\lambda+x)$. Then,
	\begin{enumerate}
		\item $Y$ can be written as
		\begin{equation}\label{eq:specform}
		Y(y)=y_0 \frac{\exp(g(y))(q/e)^y}{y!}, \quad y=0,1,2,...,
		\end{equation}
		for any constants $y_0,q>0$ and some $g$ satisfying $g(y)\leq y\log y + O(y)$ when $y\to\infty$. Moreover, we can always choose $q\in(0,1]$ and $g(y)\leq y\log y +o(y)$ simultaneously.
		
		\item If $Y$ satisfies~\eqref{eq:specform} for some $y_0$, $q$, and $g$, then
		\begin{equation}\label{eq:specKL}
		\KL(Y_x||Y)=-\log y_0 -\mathds{E}[Y_x]\log q + (\lambda+x)\log (\lambda+x)-e^{-(\lambda+x)}G(\lambda+x)
		\end{equation}
		for all $x\geq 0$;
%
		
		\item Suppose $X$ is capacity-achieving for the DTP channel with dark current $\lambda$ under an average-power constraint $\mu$ and peak-power constraint $A$ (we may have $A=\infty$). Furthermore, let $Y$ be the associated output distribution. Then, we can choose $y_0$, $q$, and $g$ in~\eqref{eq:specform} such that
		\begin{equation}\label{eq:condG}
			G(\lambda+x)\geq (\lambda+x)e^{\lambda+x}\log(\lambda+x),\quad \forall x\in [0,A]
		\end{equation}
		 with equality for all $x\in\mathsf{supp}(X)$.
	\end{enumerate}
\end{lem}
\begin{proof}
	We begin with the first point. Fix $y_0,q>0$, and consider $g$ defined as
	\[
		g(y)=\log y! + y - y\log q -\log y_0+\log Y(y).
	\]
	It is clear that
	\[
		Y(y)=y_0 \frac{\exp(g(y))(q/e)^y}{y!},
	\]
	for all $y\geq 0$. Moreover, we have $-\infty<\log Y(y)<0$. Thus, it follows that $g$ is defined and
	\[
		g(y)<\log y! + y - y\log q -\log y_0=y\log y + O(y),
	\]
	as desired. It remains to see that we can actually have $q\in (0,1]$ and $g(y)\leq y\log y+o(y)$ at the same time. This follows immediately from the observation that, if $q=1$, then
	\[
		g(y)=\log y! + y -\log y_0+\log Y(y)\leq y\log y +\frac{1}{2}\log y + O(1)=y\log y+o(y).
	\]
	
	
	For the second point, write $Y$ as in~\eqref{eq:specform}. Then, noting that $Y_x=\mathsf{Poi}(\lambda+x)$,
	\begin{align*}
		\KL(Y_x||Y)&=-H(Y_x)-\mathds{E}[\log Y(Y_x)]\\
		&= (\lambda+x)(\log(\lambda+x)-1)-\mathds{E}[\log Y_x!]-\mathds{E}[\log y_0 + g(Y_x)+Y_x\log q -Y_x-\log Y_x!]\\
		&= -\log y_0-\mathds{E}[Y_x]\log q +(\lambda+x)\log(\lambda+x)-\mathds{E}[g(Y_x)],
	\end{align*}
	with the convention that $0\log 0=0$. The result follows by observing that $\mathds{E}[g(Y_x)]=e^{-(\lambda+x)}G(\lambda+x)$.
	
	Regarding the third point, let $X$ be as in the theorem statement, and let $\mu_X=\mathds{E}[X]$. In particular, $X$ is capacity-achieving among all input distributions with support contained in $[0,A]$ and expected value at most $\mu_X$. Equivalently, $X$ is capacity-achieving among all input distributions with support contained in $[0,A]$ and output expected value at most $\mu_X+\lambda$. According to Theorem~\ref{thm:dual}, we know there exist $a\in\mathbb{R}^{\geq 0}$ and $b\in\mathbb{R}$ such that
	\begin{equation}\label{eq:KLineqopt}
		\KL(Y_x||Y)\leq a\mathds{E}[Y_x]+b
	\end{equation}
	for all $x\in [0,A]$, with equality if $x\in \mathsf{supp}(X)$.
	
	Choose $y_0=e^{-b}$ and $q=e^{-a}$. Then, there is $g$ satisfying $g(y)\leq y\log y+O(y)$ and such that~\eqref{eq:specform} holds for $Y$ with these choices of $y_0$ and $q$. According to~\eqref{eq:specKL}, we have
	\begin{equation}\label{eq:KLopt}
		\KL(Y_x||Y)=a\mathds{E}[Y_x]+b+(\lambda+x)\log (\lambda+x)-\mathds{E}[g(Y_x)].
	\end{equation}
	Note that $\mathds{E}[g(Y_x)]=e^{-(\lambda+x)}G(\lambda+x)$. Then, from~\eqref{eq:KLineqopt} and~\eqref{eq:KLopt} it follows that
	\[
		\mathds{E}[g(Y_x)]-(\lambda+x)\log (\lambda+x)=e^{-(\lambda+x)}G(\lambda+x)-(\lambda+x)\log (\lambda+x)\geq 0
	\]
	with equality for all $x\in \mathsf{supp}(X)$. This concludes the proof.
\end{proof}

We will also need the following concentration bound for the Poisson distribution, which is a consequence of Bennett's inequality (see~\cite{Can17}).
\begin{lem}\label{lem:concpoisson}
	For $0\leq\delta\leq 1$, we have
	\[
		\Pr[|\mathsf{Poi}(\lambda)-\lambda|\leq \delta\lambda]\geq 1-2\exp\left(-\frac{\delta^2\lambda}{4}\right).
	\]
\end{lem}

For completeness, we now show that the support of a capacity-achieving input distribution for the DTP channel under an average-power constraint only must be unbounded. This result was originally proved in~\cite{CHC14a}. Our proof follows a similar technique to the proof in~\cite{Sha90} that the support of a capacity-achieving distribution for the DTP channel under a peak-power constraint $A<\infty$ is finite.
\begin{thm}\label{thm:unboundedsupp}
	Suppose $X$ is a capacity-achieving distribution for the DTP channel with dark current $\lambda$ under an average-power constraint $\mu$ and no peak-power constraint. Then, $\mathsf{supp}(X)$ is unbounded.
\end{thm}
\begin{proof}
	Fix $X$ as in the theorem statement, and let $Y$ be the corresponding output distribution. Furthermore, let $\mu_X=\mathds{E}[X]$. In particular, $X$ is capacity-achieving among all input distributions with expected value at most $\mu_X$, and all such input distributions are exactly those whose corresponding output distributions have expected value at most $\mu_X+\lambda$. Then, by Theorem~\ref{thm:dual} we know there exist $a,b\in\mathbb{R}$ such that
	\begin{equation}\label{eq:KLineqsupp}
		\KL(Y_x||Y)\leq a\mathds{E}[Y_x]+b
	\end{equation}
	for all $x\geq 0$.
	
	Suppose that $\mathsf{supp}(X)\subseteq [0,x_0]$ for some $x_0$. Let $F$ be the cumulative distribution function of $X$. Then, we have
	\begin{align*}
		Y(y)&=\int_{0}^{x_0}e^{-(\lambda+x)}\frac{(\lambda+x)^y}{y!}dF(x)\\
		&\leq \int_0^{x_0} e^{-\lambda}\frac{(\lambda+x_0)^y}{y!}dF(x)\\
		&=e^{-\lambda}\frac{(\lambda+x_0)^y}{y!}.
	\end{align*}
	It follows that
	\[
		-\log Y(y)\geq \log y! +\lambda-y\log(\lambda+x_0),
	\]
	and so we have
	\begin{equation}\label{eq:logy}
		-\log Y(y)\geq (1-o(1))y\log y
	\end{equation}
	 when $y\to\infty$. As a consequence,
	\begin{align}
		-\mathds{E}[\log Y(Y_x)]&=-\sum_{y=0}^\infty Y_x(y)\log Y(y)\nonumber\\
		&\geq \Pr[Y_x\geq (1-(\lambda+x)^{-1/3})(\lambda+x)](1-o(1))(\lambda+x-(\lambda+x)^{2/3})\log(\lambda+x-(\lambda+x)^{2/3})\nonumber\\
		&\geq (1-2\exp(-x^{1/3}/4))(1-o(1))(\lambda+x-(\lambda+x)^{2/3})\log(\lambda+x-(\lambda+x)^{2/3})\nonumber\\
		&\geq (1-o(1))(\lambda+x)\log(\lambda+x)\label{eq:logyopt}
	\end{align}
	when $x\to \infty$. The first inequality holds when $x\to\infty$ due to~\eqref{eq:logy}. The second inequality follows from Lemma~\ref{lem:concpoisson} with $\delta=(\lambda+x)^{-1/3}$ .
	
	On the other hand,
	\begin{equation}\label{eq:entopt}
		H(Y_x)=O(\log(\lambda +x))
	\end{equation}
	when $x\to\infty$. This holds since $H(Y_x)$ is upper bounded by the entropy of a geometric distribution with expected value $\lambda+x$, as it maximizes the entropy over all distributions over the non-negative integers with fixed expected value. Therefore, if we let $h$ denote the binary entropy function,
	\[
		H(Y_x)\leq (\lambda+x)h\left(\frac{1}{\lambda+x}\right)=O(\log(\lambda+x))
	\]
	when $x\to\infty$, as desired.
	
	From~\eqref{eq:logyopt} and~\eqref{eq:entopt} it follows that
	\[
		\KL(Y_x||Y)=-H(Y_x)-\mathds{E}[\log Y(Y_x)]=\Omega((\lambda+x)\log(\lambda+x)).
	\]
	However, if this holds there cannot be constants $a,b\in\mathbb{R}$ such that~\eqref{eq:KLineqsupp} holds, since $\mathds{E}[Y_x]=\lambda+x$. This is a contradiction, as we assumed that $Y$ was dual feasible.
\end{proof}

To conclude this section, we show that capacity-achieving input distributions for the DTP channel under an average-power constraint and/or a peak-power constraint must be discrete. We actually prove that the support of a capacity-achieving distribution $X$ under an average-power constraint and/or a peak-power constraint must have finite intersection with every bounded interval. In particular, our techniques also recover Shamai's result for the DTP channel under a peak-power constraint~\cite{Sha90} in an alternative way.
\begin{thm}\label{thm:discretesupp}
	Suppose $X$ is a capacity-achieving distribution for the DTP channel with dark current $\lambda$ under an average-power constraint $\mu>0$ and/or a peak-power constraint $A$ (we may have $A=\infty$). Then, $\mathsf{supp}(X)\cap I$ is finite for every bounded interval $I$. In particular, $\mathsf{supp}(X)$ is countably infinite when $A=\infty$ and finite when $A<\infty$.
\end{thm}
\begin{proof}
	The statement for an average-power constraint $\mu=0$ is trivial, so we assume $\mu>0$. Fix $X$ as in the theorem statement, and let $Y$ be the corresponding output distribution. Define $\mu_X=\mathds{E}[X]$. Then, $X$ is optimal over all distributions with support in $[0,A]$ and mean at most $\mu_X$ (regardless of whether there is an average-power constraint in place or not). As a result, Lemma~\ref{lem:specform} guarantees the existence of a function $g$ such that its exponential generating function $G$ satisfies
	\[
		G(\lambda+x)\geq (\lambda+x)e^{\lambda+x}\log(\lambda+x),\quad \forall x\in [0,A]
	\]
	with equality for $x\in\mathsf{supp}(X)$. Under a change of variables, this is equivalent to
	\[
		G(x)\geq xe^{x}\log x=:f(x),\quad \forall x\in [\lambda, A+\lambda],
	\]
	with equality for $x\in S=\mathsf{supp}(X)+\lambda$.
	
	Suppose there exists a bounded interval $I$ such that $\mathsf{supp}(X)\cap I$ is infinite. As a result, we have that $S'=S\cap (I+\lambda)$ is also infinite. 
	
	Since $Y$ is an output distribution of the DTP channel and $\mathds{E}[X]>0$ necessarily (otherwise $I(X;Y)=0$ and $X$ is not capacity-achieving), we have that $Y$ has full support. Combining this with the fact that $-\log Y(y)=O(y\log y)$ when $y\to\infty$ for any output distribution $Y$ of the DTP channel, it follows that $\KL(Y_x||Y)$ is finite for every $x\geq 0$. Recalling~\eqref{eq:KLopt}, we have that $G(z)$ is finite for every $z\geq \lambda$, and hence for every $z\in\mathbb{R}$. Therefore, since $G$ is a power series, it follows that $G$ is real analytic in $(-\infty,\infty)$. Moreover, we have that $f$ is real analytic in $(0,\infty)$. 
	
	Since $G$ and $f$ are both real analytic in $(0,\infty)$ and agree on an infinite set $S'$ in this interval, it follows that $G(x)=f(x)$ for all $x\in(0,\infty)$ provided that $S'$ has a limit point in $(0,\infty)$ (via the identity theorem for real analytic functions~\cite[Corollary 1.2.6]{KP02}).
	
	Assume that indeed $S'$ has a limit point in $(0,\infty)$. Then, it follows that $G(x)=f(x)$ for all $x\in(0,\infty)$. We show that this leads to a contradiction. In fact, note that, according to \eqref{eq:expgen}, it follows that $G$ is real analytic with finite $i$-th derivative $g(i)$ at $x=0$. On the other hand, the first right-derivative of $f$ at $x=0$ is infinite. This means that we cannot have $G(x)=f(x)$ for $0<x< \infty$. As a result, we conclude that $\mathsf{supp}(X)\cap I$ must be finite, as desired.
	
	We now prove that $S'$ must have a limit point in $(0,\infty)$. Suppose that $S'$ has no limit points in $(0,\infty)$. Then, since $S'$ is a bounded infinite set, it must be the case that $0$ is a limit point of $S'$ (bounded infinite sets have at least one limit point). We show that $0$ cannot be a limit point of $S'$. If $\lambda>0$ this is trivially true since $S'\subseteq I+\lambda$ and so its limit points are at least as large as $\lambda$. We therefore assume $\lambda=0$.
	
	Suppose that $0$ is a limit point of $S'$. Then, there exists a sequence $(x_i)$ such that $x_i\in S'$ and $x_i\neq 0$ for all $i$, and $x_i\to 0$. In particular, we have $G(x_i)=f(x_i)$ for all $i$. We prove that this cannot hold. Observe that
	\[
		\lim_{i\to\infty} f(x_i)/x_i =\lim_{i\to\infty} e^{x_i}\log x_i=-\infty.
	\]
	On the other hand, recalling~\eqref{eq:expgen},
	\[
		G(x_i)/x_i=g(0)/x_i + g(1)+o(1),
	\]
	when $i\to\infty$ (and hence $x_i\to 0$). Recalling~\eqref{eq:condG} with $x=0$, we must have $G(0)=g(0)\geq 0$. As a result, $G(x_i)/x_i$ is bounded from below by a constant for $i$ large enough, and so it must be the case that $G(x_i)\neq f(x_i)$ for $i$ large enough. Therefore, $0$ cannot be a limit point of $S'$.
	
	The proof concludes by noting that
	\[
		\mathsf{supp}(X)=\bigcup_{i=0}^{A-1}(\mathsf{supp}(X)\cap [i,i+1]).
	\]
	If $A$ is finite, then so is $\mathsf{supp}(X)$. On the other hand, if $A=\infty$, then $\mathsf{supp}(X)$ is countable, and thus countably infinite by invoking Theorem~\ref{thm:unboundedsupp}.
	
	
%
\end{proof}

\section*{Acknowledgments}

The authors would like to thank Shlomo Shamai for asking them whether the capacity-achieving input distribution for the DTP channel under an average-power constraint must be discrete. This led them to the result of Section~\ref{sec:shape} that answers the question in the affirmative.

\bibliographystyle{IEEEtran}
\bibliography{poisson}

\begin{appendices}

\section{Existence of capacity-achieving distributions for the DTP channel under an average-power constraint}\label{sec:existoptimal}
In this section, we argue that capacity-achieving distributions exist for the DTP channel under an average-power constraint. For simplicity, we will assume that $\lambda=0$. Every result holds for arbitrary $\lambda\geq 0$ and under an additional peak-power constraint as well. Our approach follows that of~\cite[Appendix I]{AFTS01} closely.

We need to deal with the weak-* topology on the set $\mathcal{F}$ of probability distributions on $\mathbb{R}$. We do not define the associated concepts of weak-* compactness and continuity, but instead refer the reader to~\cite[Section 5.10]{Lue97} and the introduction of~\cite[Appendix I]{AFTS01} for the relevant background. We focus only on the parts where our approach necessarily differs from that of~\cite[Appendix I]{AFTS01}.

Let
\[
	\Omega_\mu=\left\{F\in \mathcal{F}: F(0^-)=0,\int_0^\infty x dF(x)\leq \mu\right\}.
\]
In other words, $\Omega_\mu$ is the set of probability distributions with support in $\mathbb{R}^{\geq 0}$ and bounded expected value. We also define
\[
\Omega^=_\mu=\left\{F\in \mathcal{F}: F(0^-)=0,\int_0^\infty x dF(x)= \mu\right\}.
\]

Given some distribution $F\in\Omega_\mu$, we denote by $I(F)$ the functional which maps $F$ to the mutual information $I(X_F;Y_F)$, where $X_F$, distributed according to $F$, is the input to the DTP channel and $Y_F$ is the corresponding output distribution. Then, we can write
\begin{equation}\label{eq:capsup}
	C(\mu)=\sup_{F\in\Omega_\mu}I(F).
\end{equation}

We begin by showing that capacity-achieving distributions exist for every $\mu\geq 0$. In other words, the supremum in~\eqref{eq:capsup} is actually a maximum. In order to see this, we employ the following general lemma.
\begin{lem}[\protect{\cite[Theorem 1]{AFTS01}}]\label{lem:max}
	If $J:\Omega\to\mathbb{R}$ is weak-* continuous on a weak-* compact set $\Omega\subseteq X$, where $X$ is a linear vector space, then $J$ achieves its maximum in $\Omega$.
\end{lem}
In our case, $J$ is the mutual information $I(\cdot)$ and $\Omega=\Omega_\mu$. It is easy to see that $\Omega_\mu$ is convex. One can then follow the same reasoning as in~\cite[Appendix I.A]{AFTS01} to show that $\Omega_\mu$ is weak-* compact. It remains to show that $I(\cdot)$ is weak-* continuous.
\begin{lem}
	$I(\cdot)$ is weak-* continuous in $\Omega_\mu$.
\end{lem}
\begin{proof}
	For $F\in\Omega_\mu$, we have
	\[
		I(F)=H(Y_F)-\int_0^\infty H(Y_x)dF(x),
	\]
	where $Y_F$ is the output distribution induced by $F$. We show that the two terms in the right hand side are weak-* continuous.
	
	The proof that $\int_0^\infty H(Y_x)dF(x)$ is a weak-* continuous function of $F\in\Omega_\mu$ follows in the same way as the analogous result in~\cite[Appendix I.B]{AFTS01}. This is because, for the DTP channel, $H(Y_x)$ is continuous in $x$ for all $x\geq 0$, $H(Y_x)=O(\log(1+x))$, and we have the constraint $\int_0^\infty xdF(x)\leq \mu$ for all $F\in\Omega_\mu$.
	
	It remains to show that $H(Y_F)=-\sum_{y=0}^\infty Y_F(y)\log Y_F(y)$ is a weak-* continuous function of $F\in\Omega_\mu$. Fix a sequence of distributions $(F_n)$ that converges weakly to some $F$, denoted by $F_n\xrightarrow{w} F$. In order to show that a functional $f$ is weak-* continuous, in this case it suffices to show that $f(F_n)\to f(F)$ in the usual Euclidean metric in $\mathbb{R}$ as $n\to\infty$.
	
	We have
	\begin{align}
		\lim_{n\to \infty} H(Y_{F_n})&=-\lim_{n\to \infty} \sum_{y=0}^\infty Y_{F_n}(y)\log Y_{F_n}(y)\label{eq:hy1}\\
		&=-\sum_{y=0}^\infty \lim_{n\to \infty} Y_{F_n}(y)\log Y_{F_n}(y)\label{eq:hy2}\\
		&=-\sum_{y=0}^\infty Y_F(y)\log Y_F(y)\label{eq:hy3}\\
		&=H(Y_F).\label{eq:hy4}
	\end{align}
	We justify all of the steps above. Observe that~\eqref{eq:hy1} and~\eqref{eq:hy4} follow by definition. To show~\eqref{eq:hy3}, note that, for fixed $y$, the function $x\mapsto Y_x(y)$ is a bounded, continuous function of $x$. Therefore, by the properties of the weak-* topology, it follows that
	\[
		Y_F(y)=\int_0^\infty Y_x(y)dF(x)
	\]
	is a continuous function of $F$ for each $y$. Since $x\mapsto x\log x$ is continuous for $x\geq 0$ as well,~\eqref{eq:hy3} holds. 
	
	It remains to prove~\eqref{eq:hy2}. It suffices to show that we are in a condition to apply the dominated convergence theorem. More specifically, we need to prove that
	\[
		|Y_F(y)\log Y_F(y)|\leq g(y)
	\]
	for all $F\in\Omega_\mu$ and $y\in\mathbb{N}$, where $g$ satisfies $\sum_{y=0}^\infty g(y)<\infty$. Fix $F\in\Omega_\mu$, and note that
	\begin{align}
		Y_F(y)&=\int_0^\infty e^{-x}\frac{x^y}{y!}dF(x)\nonumber\\
		&=\int_0^{y-y^{0.99}} e^{-x}\frac{x^y}{y!}dF(x)+\int_{y-y^{0.99}}^{y+y^{0.99}} e^{-x}\frac{x^y}{y!}dF(x)+\int_{y+y^{0.99}}^\infty e^{-x}\frac{x^y}{y!}dF(x).\label{eq:yfsep}
	\end{align}
	We analyze the three terms. First, since $x\mapsto Y_x(y)$ is increasing for $x<y$ and decreasing for $x>y$, we have
	\begin{align}
		&\int_0^{y-y^{0.99}} e^{-x}\frac{x^y}{y!}dF(x)\leq Y_{y-y^{0.99}}(y),\label{eq:lower1}\\
		&\int_{y+y^{0.99}}^\infty e^{-x}\frac{x^y}{y!}dF(x)\leq Y_{y+y^{0.99}}(y)\label{eq:upper3},
	\end{align}
	and both $Y_{y-y^{0.99}}(y)$ and $Y_{y+y^{0.99}}(y)$ converge to $0$ faster than $y^{-3/2}$ when $y\to\infty$.
	
	For fixed $y$, it can be seen that $Y_x(y)$ is maximized when $x=y$. Furthermore, we have $Y_y(y)=O(1/\sqrt{y})$. Since $F\in\Omega_\mu$, we have $1-F(x)\leq \mu/x$, and so
	\begin{equation}\label{eq:middle2}
		\int_{y-y^{0.99}}^{y+y^{0.99}} e^{-x}\frac{x^y}{y!}dF(x)\leq \frac{\mu Y_y(y)}{y-y^{0.99}}=O(y^{-3/2}).
	\end{equation}
	Combining~\eqref{eq:yfsep} with~\eqref{eq:lower1},~\eqref{eq:upper3}, and~\eqref{eq:middle2} yields
	\begin{equation}\label{eq:asympyf}
		Y_F(y)=O(y^{-3/2})
	\end{equation}
	 when $y\to \infty$ for all $F\in\Omega_\mu$, where the hidden constant is independent of $F$. To conclude, observe that, due to~\eqref{eq:asympyf}, for every $\eps$ there is a constant $y_\eps$ (possibly depending on $\mu$) such that $Y_F(y)\leq \eps$ for all $y\geq y_\eps$ and $F\in\Omega_\mu$. Therefore,
	\begin{align*}
		|Y_F(y)\log Y_F(y)|=O(Y_F(y)^{0.7})=O(y^{-1.05})
	\end{align*}
	for all $F\in\Omega_\mu$, where we used~\eqref{eq:asympyf} in the last equality, the hidden constant is independent of $F$. Consequently,~\eqref{eq:hy2} follows by noting that $\sum_{y=0}^\infty y^{-1.05}<\infty$. This shows that $H(Y_F)$ is weak-* continuous, and hence $I(F)$ is weak-* continuous too, as desired.
\end{proof}

Finally, Lemma~\ref{lem:max} implies that for every $\mu\geq 0$ there exists $F^\star\in\Omega_\mu$ such that
\[
	C(\mu)=I(F^\star).
\]
We can show more: If $F^\star\in\Omega_\mu$ is capacity-achieving, then $F^\star\in\Omega^=_\mu$ necessarily. In fact, suppose not, and let $\mu'=\mathds{E}_{F_0}[X]$. We have $\mu'<\mu$ by hypothesis. Then, it is clear that $C(\mu'')=I(F^\star)$ for all $\mu''\in[\mu',\mu]$. It is easy to see that $C(\mu)$ is concave in $\mu$. As a result, we have $C(\mu'')=I(F^\star)$ for all $\mu''\geq \mu'$. However, it can be shown that $C(\mu)$ is unbounded when $\mu\to \infty$. This is a contradiction, and so $F^\star\in\Omega^=_\mu$ necessarily.

\section{Proof of Theorem~\ref{thm:dual}}\label{sec:proof1}
In this section, we prove Theorem~\ref{thm:dual} for ``well-behaved" channels. The technical meaning of ``well-behaved" will be made clear later on.

Since our input alphabet is continuous, we have to deal with input distributions that do not have associated probability density/mass functions. In fact, the input distribution may be a mixture of discrete and continuous distributions. Because of this, we are forced to work solely with cumulative distribution functions, which we may call just ``distributions". Our output alphabet is discrete, and so we may identify distributions with the corresponding probability mass functions. Overall, this leads to a more technical proof, although the methods used are still standard. Our approach mimics in part those of~\cite{Smi71,Sha90,AFTS01}. Additionally, we present proofs of standard results whose proofs we could not find in the literature. 

We note that if we deal with discrete inputs only, then the proof of the analogous result in this case is shorter~\cite{Che17}, but leads to the exact same conclusions.

Before we proceed with the proof of Theorem~\ref{thm:dualtech}, we need some auxiliary definitions and results. Given a functional $f\colon \Omega\to\mathbb{R}$, where $\Omega$ is a convex subset of a linear vector space, the \emph{weak derivative of $f$ at $F\in\Omega$ in the direction of $Q\in\Omega$}, denoted by $f'_F(Q)$, is defined as
\[
	f'_F(Q)=\lim_{\theta\to 0^+}\frac{f((1-\theta)F+\theta Q)-f(F)}{\theta}.
\]
The functional $f$ is said to be \emph{weakly differentiable in $\Omega$ at $F$} if $f'_F(Q)$ exists for all $Q\in\Omega$. If $f$ is weakly differentiable in $\Omega$ at $F$ for all $F\in\Omega$, then we simply say $f$ is \emph{weakly differentiable in $\Omega$}. We have the following result.
\begin{lem}\label{lem:weakder}
	Fix a concave $f\colon\Omega\to\mathbb{R}$ in a convex space $\Omega$, and suppose that $f$ achieves a maximum in $\Omega$. If $F^\star\in\Omega$ is a maximizer of $f$ in $\Omega$ and $f'_{F^\star}(Q)$ exists, then
	\[
		f'_{F^\star}(Q)\leq 0.
	\]
	Moreover, if $f$ is weakly differentiable in $\Omega$ at $F$ and $f'_{F}(Q)\leq 0$ for all $Q\in\Omega$, then $F$ maximizes $f$ in $\Omega$.
\end{lem}
\begin{proof}
	Fix $f$ satisfying the conditions of the lemma statement, and let $F^\star\in\Omega$ be a maximizer of $f$ over $\Omega$. Therefore,
	\[
		\frac{f((1-\theta)F^\star+\theta Q)-f(F^\star)}{\theta}\leq 0
	\]
	for every $\theta\in(0,1]$, since $(1-\theta)F^\star+\theta Q\in\Omega$ by the convexity of $\Omega$ and $f(F^\star)\geq f(F)$ for every $F\in\Omega$ by hypothesis. As a result, if $f'_{F^\star}(Q)$ exists, then we must have $f'_{F^\star}(Q)\leq 0$.
	
	For the second statement, suppose that $F$ is not a maximizer. Then, there exists $Q\in\Omega$ such that $f(Q)>f(F)$. For every $\theta\in (0,1]$, we have
	\[
		\frac{f((1-\theta)F+\theta Q)-f(F)}{\theta}\geq \frac{(1-\theta)f(F)+\theta f(Q)-f(F)}{\theta}=f(Q)-f(F)>0,
	\]
	where the first inequality follows from the concavity of $f$. Since this result holds for every $\theta\in(0,1]$, we conclude that $f'_F(Q)\geq f(Q)-f(F)>0$. Therefore, if $f'_F(Q)\leq 0$ for all $Q$, then $F$ must be a maximizer of $f$ in $\Omega$.
\end{proof}

The following lemma states a generalized form of Lagrange duality. Informally, this result transforms a constrained convex optimization problem (such as determining the capacity of a channel under some average-power constraint) into an unconstrained optimization problem. This is accomplished by moving the constraint into the objective function to be optimized. The necessary and sufficient conditions for optimality of a candidate solution to the unconstrained problem have a more useful form, as we shall see later in this section.
\begin{lem}[{\cite[Section 8.6, Theorem 1, specialized]{Lue97}}]
\label{lem:lagdual}
	Let $f\colon\Omega\to\mathbb{R}$ be convex, where $\Omega$ is a convex subset of a vector space $X$, and let $G:\Omega\to\mathbb{R}$ be a convex map. Suppose there exists an $x\in\Omega$ such that $G(x)<0$, and that $\inf\{f(x):G(x)\leq 0,x\in\Omega\}$ is finite. Then,
	\[
		\inf\{f(x):G(x)\leq 0,x\in\Omega\}=\max\{\varphi(z):z\geq 0\},
	\]
	where $\varphi(z)=\inf\{f(x)+zG(x):x\in\Omega\}$, and the maximum on the right hand side is achieved by some $z^\star$.
	
	Moreover, if the infimum on the left hand side is achieved by some $x^\star$, then
	\[
		z^\star G(x^\star)=0,
	\]
	and $x^\star$ minimizes $f(x)+z^\star G(x)$ over $\Omega$.
\end{lem}

Given some channel $\Ch$ with input alphabet $\mathcal{X}$ and output alphabet $\mathcal{Y}$, we can define the associated mutual information functional $I(\cdot)$. Suppose that the output distribution of $\Ch$ given input $x\in\mathcal{X}$ has an associated probability density function $Y_x(\cdot)$. Given a distribution $F$ on $\mathcal{X}$, we define $I(F)=I(X_F;Y_F)$, where $X_F$ is an input distribution to $\Ch$ distributed according to $F$ and $Y_F$ is the corresponding output distribution satisfying
\[
	Y_F(y)=\int_\mathcal{X}Y_x(y)dF(x),\quad\forall y\in\mathcal{Y}.
\]

The following result characterizes the weak derivative of $I(\cdot)$, conditioned on a certain quantity being finite. This characterization, combined with Lemma~\ref{lem:weakder}, is the key to determining the conditions under which an input distribution is capacity-achieving.
\begin{lem}\label{lem:weakderexp}
	Let $I(\cdot)$ denote the mutual information functional of some channel $\Ch$ with input alphabet $\mathcal{X}\subseteq \mathbb{R}^{\geq 0}$ and output alphabet $\mathcal{Y}\subseteq\mathbb{N}$. Fix an input distribution $F$ on $\mathcal{X}$ such that $I(F)<\infty$ with corresponding output distribution $Y_F$. Suppose that
	\[
		\int_0^\infty \KL(Y_x||Y_F)dQ(x)<\infty
	\]
	for a distribution $Q$ on $\mathcal{X}$. Then, $I'_F(Q)$ exists and is given by
	\begin{equation}\label{eq:weakderexp}
	I'_F(Q)=\int_0^\infty \KL(Y_x||Y_F)dQ(x) - I(F).
	\end{equation}
\end{lem}
\begin{proof}
	Let $F_\theta=(1-\theta)F+\theta Q$ for $\theta\in[0,1]$. Denote the output distribution associated to $F_\theta$ by $Y_\theta=(1-\theta)Y_0+\theta Y_1$, where $Y_0$ and $Y_1$ denote the output distributions of $F$ and $Q$, respectively. We have
	\begin{align}
		\frac{I(F_\theta)-I(F)}{\theta}&=\int_0^\infty \KL(Y_x||Y_\theta)dQ(x)-I(F)\nonumber\\&+\frac{1-\theta}{\theta}\left(\int_0^\infty \KL(Y_x||Y_\theta)dF(x)-\int_0^\infty\KL(Y_x||Y_0)dF(x)\right).\label{eq:expandweakder}
	\end{align}
	We deal with the limit of each term on the right hand side of~\eqref{eq:expandweakder} separately. First, we show that
	\begin{equation}\label{eq:limterm1}
		\lim_{\theta\to 0^+}\int_0^\infty \KL(Y_x||Y_\theta)dQ(x)=\int_0^\infty \KL(Y_x||Y_0)dQ(x).
	\end{equation}
	Observe that
	\begin{align}
		\KL(Y_x||Y_\theta)&=-H(Y_x)-\sum_{y=0}^\infty Y_x(y)\log((1-\theta)Y_0(y)+\theta Y_1(y))\nonumber\\
		&\leq -H(Y_x)-\sum_{y=0}^\infty Y_x(y)\log((1-\theta)Y_0(y))\nonumber\\
		&=\KL(Y_x||Y_0)-\log(1-\theta)\label{eq:ineqKL1}
	\end{align}
	for all $\theta\in [0,1)$. Since $\int_0^\infty \KL(Y_x||Y_0)dQ(x)<\infty$ by hypothesis, we conclude from~\eqref{eq:ineqKL1} that
	\[
		\int_0^\infty \KL(Y_x||Y_\theta)dQ(x)<\infty
	\]
	for all $\theta\in [0,1)$. Hence, by Fubini's theorem,
	\begin{align}
		\int_0^\infty \KL(Y_x||Y_\theta)dQ(x)&=-\int_0^\infty \left(H(Y_x)+\sum_{y=0}^\infty Y_x(y)\log Y_\theta(y)\right)dQ(x)\nonumber\\
		&=-H(Y|X_Q)-\sum_{y=0}^\infty \int_0^\infty Y_x(y)\log Y_\theta(y) dQ(x)\nonumber\\
		&=-H(Y|X_Q)-\sum_{y=0}^\infty Y_1(y)\log Y_\theta(y).\label{eq:term1}
	\end{align}
	Therefore,
	\begin{equation}\label{eq:simp1}
		\lim_{\theta\to 0^+}\int_0^\infty \KL(Y_x||Y_\theta)dQ(x)=-H(Y|X_Q)-\lim_{\theta\to 0^+}\sum_{y=0}^\infty Y_1(y)\log Y_\theta(y).
	\end{equation}
	We now show that we can swap the limit and infinite sum on the right hand side of~\eqref{eq:simp1}. Observe that
	\[
		-\log Y_\theta(y)\leq -\log Y_0(y)-\log(1-\theta)\leq -\log Y_0(y)+2,
	\]
	provided that $\theta$ is small enough. Since $-\sum_{y=0}^\infty Y_1(y)\log Y_0(y)<\infty$ by hypothesis (recall we assume $\int_0^\infty \KL(Y_x||Y_0)dQ(x)<\infty$), it follows by the dominated convergence theorem that
	\begin{equation}
		-\lim_{\theta\to 0^+}\sum_{y=0}^\infty Y_1(y)\log Y_\theta(y)=-\sum_{y=0}^\infty Y_1(y)\lim_{\theta\to 0^+}\log Y_\theta(y)=-\sum_{y=0}^\infty Y_1(y)\log Y_0(y). \label{eq:domconv}
	\end{equation}
	Combining~\eqref{eq:domconv} with~\eqref{eq:simp1} yields~\eqref{eq:limterm1}, as desired.
	
	We now show that
	\begin{equation}\label{eq:rightder}
	\lim_{\theta\to 0^+}\frac{1-\theta}{\theta}\int_0^\infty (\KL(Y_x||Y_\theta)-\KL(Y_x||Y_0))dF(x)=0.
	\end{equation}
	The limit on the left hand side of~\eqref{eq:rightder} equals the right derivative of $\int_0^\infty \KL(Y_x||Y_\theta)dF(x)$ with respect to $\theta$ at $\theta=0$. We show that this limit is zero by a reasoning similar to the proof of Leibniz's integral rule. First, from~\eqref{eq:ineqKL1} it follows that
	\[
	\int_0^\infty \KL(Y_x||Y_\theta)dF(x)=-H(Y|X_F)+\int_0^\infty\sum_{y=0}^\infty Y_x(y)\log\left(\frac{1}{Y_\theta(y)}\right)dF(x)\leq I(F)-\log(1-\theta)<\infty
	\]
	for $\theta\in[0,1)$. As a result, by Fubini's theorem we have
	\begin{equation}\label{eq:simp3}
	\int_0^\infty\sum_{y=0}^\infty Y_x(y)\log\left(\frac{1}{Y_\theta(y)}\right)dF(x)=\sum_{y=0}^\infty\int_0^\infty Y_x(y)\log\left(\frac{1}{Y_\theta(y)}\right)dF(x)=\sum_{y=0}^\infty Y_0(y)\log\left(\frac{1}{Y_\theta(y)}\right)
	\end{equation}
	for $\theta\in[0,1)$, with the convention that $0\log x=0$ for all $x\geq 0$.
	
	Let $h(\theta,y)=\log\left(\frac{1}{Y_\theta(y)}\right)$. Since $Y_\theta(y)=(1-\theta)Y_0(y)+\theta Y_1(y)$, we have
	\[
	\frac{\partial h}{\partial\theta} (\theta,y)=\frac{Y_0(y)-Y_1(y)}{(1-\theta)Y_0(y)+\theta Y_1(y)}
	\]
	for all $\theta\in (0,1)$ and $y$. In particular, $h(\cdot,y)$ is, say, continuous in $[0,1/2]$ and differentiable in $(0,1/2)$ for every $y\in\mathsf{supp}(Y_0)$. Moreover, we have the general bound
	\begin{equation}\label{eq:derbound}
	\left|\frac{\partial h}{\partial\theta} (\theta,y)\right|\leq \frac{1}{1-\theta}\left|1-\frac{Y_1(y)}{Y_0(y)}\right|\leq 2\left(1+\frac{Y_1(y)}{Y_0(y)}\right)=:g(y)
	\end{equation}
	for all $y\in\mathsf{supp}(Y_0)$, provided that $\theta<1/2$. Observe that
	\[
	\sum_{y=0}^\infty Y_0(y)g(y)=2+2\sum_{y=0}^\infty Y_1(y)=4,
	\]
	and so $g$ is integrable with respect to $Y_0$.
	
	Via~\eqref{eq:simp3}, we can write the left hand side of~\eqref{eq:rightder} as
	\[
	\lim_{\theta\to 0^+}\sum_{y=0}^\infty Y_0(y)\frac{h(\theta,y)-h(0,y)}{\theta}.
	\]
	Since $h(\cdot,y)$ is continuous in $[0,1/2]$ and differentiable in $(0,1/2)$ for every $y\in\mathsf{supp}(Y_0)$, then, by the mean value theorem, for every $\theta\in(0,1/2]$ and $y\in\mathsf{supp}(Y_0)$ there exists some $z\in (0,\theta)$ such that
	\[
	\frac{h(\theta,y)-h(0,y)}{\theta}=\frac{\partial h}{\partial\theta} (z,y).
	\]
	Taking into account~\eqref{eq:derbound}, it follows that
	\[
	\left|\frac{h(\theta,y)-h(0,y)}{\theta}\right|\leq g(y)
	\]
	for all $y\in\mathsf{supp}(Y_0)$ and $\theta\in(0,1/2]$. Therefore, by the dominated convergence theorem we can conclude that
	\[
	\lim_{\theta\to 0^+}\sum_{y=0}^\infty Y_0(y)\frac{h(\theta,y)-h(0,y)}{\theta}=\sum_{y=0}^\infty Y_0(y)\lim_{\theta\to 0^+}\frac{h(\theta,y)-h(0,y)}{\theta}=\sum_{y=0}^\infty Y_0(y)\left(1-\frac{Y_1(y)}{Y_0(y)}\right)=1-1=0,
	\]
	which shows that~\eqref{eq:rightder} holds.
	
	Finally, combining~\eqref{eq:expandweakder},~\eqref{eq:limterm1}, and~\eqref{eq:rightder} yields the desired result.
\end{proof}

The following is a generalization to continuous alphabets of a well-known convex duality result for discrete memoryless channels~\cite[Chapter 2, Theorem 3.4]{CK11}. We use it to derive our general capacity upper bound in an easy way.
\begin{lem}[\protect{\cite[Theorem 5.1, specialized]{LM03}}]\label{lem:lapi}
	Fix a channel $\Ch$ with input alphabet $\mathcal{X}\subseteq \mathbb{R}^{\geq 0}$ and output alphabet $\mathcal{Y}\subseteq\mathbb{N}$. Suppose that for every set $S\subseteq\mathbb{N}$ the map $x\mapsto Y_x(S)=\sum_{y\in S} Y_x(y)$ is Borel-measurable. Let $F$ be any distribution on $\mathcal{X}$, and $Y$ any distribution on $\mathcal{Y}$. Then,
	\[
	I(F)\leq \int_0^\infty \KL(Y_x||Y)dF(x).
	\]
\end{lem}

Finally, we define some sets which will be relevant in the proof.
\begin{align*}
\OYmu&=\left\{F\in\mathcal{F}:\mathsf{supp}(F)\subseteq [0,A],\mathds{E}[Y_F]=\int_0^\infty\mathds{E}[Y_x]dF(x)\leq \mu\right\},\\
\OYeq&=\left\{F\in\mathcal{F}:\mathsf{supp}(F)\subseteq [0,A],\mathds{E}[Y_F]=\int_0^\infty\mathds{E}[Y_x]dF(x)= \mu\right\},\\
\OYfin&= \left\{F\in\mathcal{F}:\mathsf{supp}(F)\subseteq [0,A],\mathds{E}[Y_F]=\int_0^\infty\mathds{E}[Y_x]dF(x)< \infty\right\}.
\end{align*}
Note that all of these sets are convex subsets of the set $\mathcal{F}$ of all distributions on $\mathbb{R}$. We are now ready to prove the following theorem.
\begin{thm}[Theorem~\ref{thm:dual}, generalized]\label{thm:dualtech}
	Let ${\sf Ch}$ be a channel with input alphabet $\mathcal{X}=[0,A]$ (where we may set $A=\infty$) and output alphabet $\mathbb{N}$. Furthermore, let $\Ch_\mu$ denote $\Ch$ under an output average-power constraint $\mu$. Suppose that, for every $S\subseteq \mathbb{N}$, the map $x\mapsto Y_x(S)=\sum_{y\in S} Y_x(y)$ is Borel-measurable. Then, we have the following:
	\begin{enumerate}
		\item Assume that there exist a random variable $Y$, supported on $\mathcal{Y}$, and parameters $\nu_0\in\R$ and $\nu_1\in\mathbb{R}^{\geq 0}$ such that
		\begin{equation*}
		D_{{\sf KL}}(Y_x\| Y)\leq \nu_0+\nu_1\E[Y_x]
		\end{equation*}
		for every $x\in\mathcal{X}$. Then, we have
		\[
		C({\sf Ch}_\mu)\leq \nu_0+\nu_1\mu
		\]
		for every $\mu$ such that $\OYmu\neq\emptyset$.
		
		\item 
		Suppose that $\KL(Y_x||Y_F)$ exists for all $x\in\mathcal{X}$ and all output distributions $Y_F$ associated to input distributions $F$ satisfying $\mathds{E}[X_F]>0$, that the map $x\mapsto \KL(Y_x||Y)$ is continuous in $x$, that for each $\mu>0$ there is $F$ such that $\mathds{E}[Y_F]<\mu$, and that $x\mapsto\mathds{E}[Y_x]$ is continuous in $x$. Then, if $F^\star\in\OYmu$ with $I(F^\star)<\infty$ is capacity-achieving for $\Ch_\mu$ we must have
		\begin{equation}\label{eq:KLline}
			\KL(Y_x||Y_{F^\star})\leq \nu_0+\nu_1\mathds{E}[Y_x],\quad\forall x\in\mathcal{X}
		\end{equation}
		for some $\nu_0\in\mathbb{R}$ and $\nu_1\in\mathbb{R}^{\geq 0}$, with equality for $x\in\mathsf{supp}(F^\star)$. Moreover, if $F\in\OYeq$ satisfies~\eqref{eq:KLline} for some $\nu_0\in\mathbb{R}$ and $\nu_1\in\mathbb{R}^{\geq 0}$ with equality for $x\in\mathsf{supp}(F)$, then $F$ is capacity-achieving for $\Ch_\mu$ and the capacity in this case is exactly
		\[
		C(\Ch_\mu)=\nu_0+\nu_1\mu.
		\]
		
		
	\end{enumerate}
\end{thm}
\begin{proof}
	We begin by proving the first part of the theorem statement. Fix some distribution $Y$ in $\mathcal{Y}$ such that
	\begin{equation}\label{eq:condY}
		\KL(Y_x||Y)\leq \nu_0+\nu_1\mathds{E}[Y_x],\quad\forall x\in\mathcal{X}.
	\end{equation}
	Furthermore, let $F\in\OYmu$ be some input distribution. By Lemma~\ref{lem:lapi}, we have
	\[
		I(F)\leq \int_0^\infty \KL(Y_x||Y)dF(x)\leq \nu_0+\nu_1\int_0^\infty \mathds{E}[Y_x]dF(x)\leq \nu_0+\nu_1\mu.
	\]
	The first inequality follows from Lemma~\ref{lem:lapi}. The second inequality follows from~\eqref{eq:condY}. Finally, the third inequality holds because $F\in\OYmu$. This implies that $C(\Ch_\mu)\leq \nu_0+\nu_1\mu$, as desired.
	
	We now prove the second part of the theorem statement. First, suppose that $F^\star\in\OYmu$ is capacity-achieving for $\Ch_\mu$. Instantiate Lemma~\ref{lem:lagdual} with $\Omega=\OYfin$, $f(F)=-I(F)$, and $G(F)=\mathds{E}[Y_F]-\mu$. By hypothesis, we have that $I(F^\star)<\infty$, and that there exists $F$ with $G(F)<0$ whenever $\mu>0$. Moreover, both $-I$ and $G$ are convex, and $\OYfin$ is a convex subspace of a vector space. As a result, there exists $z^\star\geq 0$ such that $F^\star$ minimizes
	\[
		J(\cdot)=-I(\cdot)+z^\star G(\cdot)
	\]
	over $\OYfin$, and $z^\star G(F^\star)=0$. As a result, according to Lemma~\ref{lem:weakder} we must have
	\begin{equation}\label{eq:condweakder}
		J'_{F^\star}(Q)=I'_{F^\star}(Q)-z^\star G'_{F^\star}(Q)\leq 0
	\end{equation}
	if $I'_{F^\star}(Q)$ and $G'_{F^\star}(Q)$ exist. For a fixed $\overline{x}\in\mathcal{X}$, define the unit step function $Q_{\overline{x}}\in\OYfin$ as
	\[
		Q_{\overline{x}}(x)=\begin{cases}
		0,\text{ if $x<\overline{x}$}\\
		1,\text{ else.}
		\end{cases}
	\]
	Since $\KL(Y_x||Y_{F^\star})$ is finite for every $x\in\mathcal{X}$, Lemma~\ref{lem:weakderexp} implies that $I'_{F^\star}(Q_{\overline{x}})$ exists and is given by
	\begin{equation}\label{eq:weakderstep}
		I'_{F^\star}(Q_{\overline{x}})=\int_0^\infty \KL(Y_x||Y_{F^\star})dQ_{\overline{x}}(x)-I(F^\star)=\KL(Y_{\overline{x}}||Y_{F^\star})-I(F^\star).
	\end{equation}
	Furthermore, since $G$ is linear in $\OYfin$, we have
	\begin{equation}\label{eq:weakderG}
		G'_{F^\star}(Q_{\overline{x}})=\lim_{\theta\to 0^+}\frac{G(F_\theta)-G(F^\star)}{\theta}=G(Q_{\overline{x}})-G(F^\star)=\mathds{E}[Y_{\overline{x}}]-\mathds{E}[Y_{F^\star}].
	\end{equation}
	
	Combining~\eqref{eq:condweakder},~\eqref{eq:weakderstep}, and~\eqref{eq:weakderG}, we must have
	\[
		\KL(Y_{x}||Y_{F^\star})-I(F^\star)-z^\star\mathds{E}[Y_{x}]+z^\star\mathds{E}[Y_{F^\star}]\leq 0
	\]
	for every $x\in\mathcal{X}$. Equivalently,
	\[
		\KL(Y_{x}||Y_{F^\star})\leq I(F^\star)+z^\star(\mathds{E}[Y_{x}]-\mu)
	\]
	must hold for every $x\in\mathcal{X}$. The inequality holds because, according to Lemma~\ref{lem:lagdual}, if $G(F^\star)\neq 0$ (i.e., $\mathds{E}[Y_{F^\star}]<\mu$), then we must have $z^\star=0$ and the inequality would still be true in this case.
	
	Suppose now that there is $x\in\mathsf{supp}(F)$ such that
	\begin{equation}\label{eq:strictsupp}
		\KL(Y_x||Y_{F^\star})<I(F^\star)+z^\star (\mathds{E}[Y_x]-\mu).
	\end{equation}
	All terms in the inequality above are continuous in $x$ by hypothesis. As a result,~\eqref{eq:strictsupp} actually holds in an open neighborhood $U$ of $x$. Since $x\in\mathsf{supp}(F^\star)$, by definition of support we have $\int_U dF^\star(x)=\delta>0$ for some positive $\delta$. Therefore,
	\[
		I(F^\star)=\int_0^\infty \KL(Y_x||Y_{F^\star})dF^\star(x)<I(F^\star)+z^\star \left(\int_0^\infty\mathds{E}[Y_x]dF^\star(x)-\mu\right)=I(F^\star),
	\]
	a contradiction. It follows that for $\nu_1=z^\star\geq 0$ and $\nu_0=I(F^\star)-z^\star\mu$ we must have
	\[
	\KL(Y_{x}||Y_{F^\star})\leq \nu_0+\nu_1\mathds{E}[Y_x]
	\]
	for every $x$, with equality for $x\in\mathsf{supp}(F^\star)$, as desired.
	
	Finally, suppose that $F\in\OYeq$ satisfies
	\[
	\KL(Y_{x}||Y_F)\leq \nu_0+\nu_1\mathds{E}[Y_x]
	\]
	for some $\nu_0\in\mathbb{R}$ and $\nu_1\in\mathbb{R}^{\geq 0}$ and for every $x\in\mathcal{X}$, with equality for $x\in\mathsf{supp}(F)$. Then, for every distribution $Q\in\OYfin$ we have
	\[
		\int_0^\infty \KL(Y_x||Y_F)dQ(x)\leq \nu_0+\nu_1 \mathds{E}[Y_Q]<\infty.
	\]
	The last inequality follows since $Q\in \OYfin$. Furthermore, since $\KL(Y_{x}||Y_F)=\nu_0+\nu_1\mathds{E}[Y_x]$ for all $x\in\mathsf{supp}(F)$, we have
	\[
	I(F)=\int_0^\infty \KL(Y_x||Y_F)dF(x)=\nu_0+\nu_1\mu.
	\]
	The last equality holds because $F\in\OYeq$. Therefore, according to Lemma~\ref{lem:weakderexp}, $I'_F(Q)$ exists for every $Q\in \OYfin$ and satisfies
	\begin{equation}\label{eq:weakdercond}
		I'_F(Q)=\int_0^\infty \KL(Y_x||Y_F)dQ(x)-I(F)\leq \nu_0+\nu_1 \mathds{E}[Y_Q]-(\nu_0+\nu_1\mu)=\nu_1 (\mathds{E}[Y_Q]-\mu).
	\end{equation}
	As a result,
	\[
		I'_F(Q)-\nu_1 G'_F(Q)=I'_F(Q)-\nu_1(\mathds{E}[Y_Q]-\mu)\leq 0
	\]
	for every $Q\in\OYfin$. Via Lemma~\ref{lem:weakder}, it follows that $F$ minimizes the functional $J(\cdot)=-I(\cdot)+\nu_1 G(\cdot)$ over $\OYfin$. Moreover, since $\mathds{E}[Y_F]=\mu$, we have $G(F)=0$. This means that $F$ minimizes $J$ with value $J(F)=-I(F)$. If $F$ is not capacity-achieving over $\OYmu$, there exists some $F^\star\in\OYmu$ such that $I(F^\star)>I(F)$. Note that $G(F^\star)\leq 0$ and $\nu_1\geq 0$. Therefore,
	\[
		J(F^\star)=-I(F^\star)+\nu_1 G(F^\star)\leq-I(F^\star)< -I(F),
	\]
	contradicting the fact that $F$ minimizes $J$ over $\OYfin$. Therefore, we conclude that $F$ is capacity-achieving and $C(\Ch_\mu)=I(F)=\nu_0+\nu_1\mu$.
\end{proof}

\begin{remark}
	Note that, a priori, there may not exist capacity-achieving distributions for $\mathsf{Ch}_\mu$ as in the statement of Theorem~\ref{thm:dualtech}. Moreover, even if capacity-achieving distributions exist, they may not lie in $\OYeq$. It is possible to come up with stronger, but still general, assumptions about $\mathds{E}[Y_x]$ which ensure that there exist capacity-achieving distributions for $\Ch_\mu$, and that they lie in $\OYeq$. An example, which covers the DTP channel, is when $\mathds{E}[Y_x]$ is an increasing affine function of $x\in\mathcal{X}$. In this case, imposing an output average-power constraint is equivalent to imposing an input average-power constraint on the channel. As a result, the desired properties transfer directly from one setting to the other.
\end{remark}

It remains to see that the DTP channel satisfies the hypotheses of Theorem~\ref{thm:dualtech} in order to derive Theorem~\ref{thm:dual}. First, the map $x\mapsto Y_x(S)$ is continuous in $x$ for all $S\subseteq \mathbb{N}$, and hence it is Borel-measurable. Second, $\KL(Y_x||Y)$ is finite and continuous in $x$ whenever $Y$ is an output distribution of the DTP channel with full support over $\mathbb{N}$. This happens whenever $\mathds{E}[X_F]>0$. Third, the results of Appendix~\ref{sec:existoptimal} imply that $I(F)<\infty$ for every $F\in\OYfin$. In particular, this means that $C(\DTP_{\lambda,A,\mu})<\infty$ always. Fourth, observe that $\mathds{E}[Y_x]=\lambda+x$. Therefore, $x\mapsto \mathds{E}[Y_x]$ is continuous. Finally, note that we have $\mathds{E}[Y_F]=\mu$ if and only if $\mathds{E}[X_F]=\mu-\lambda$. It follows that $\OYmu\neq\emptyset$ and that there exists $F$ with $\mathds{E}[Y_F]<\mu$ whenever $\mu> \lambda$. Furthermore, because of this property, the results of Appendix~\ref{sec:existoptimal} imply that capacity-achieving distributions for the DTP channel over $\OYmu$ exist and are contained in $\OYeq$. Combining all of these observations with Theorem~\ref{thm:dualtech} leads to Theorem~\ref{thm:dual}.

\end{appendices}

\end{document}